\documentclass[12pt, draftclsnofoot, onecolumn]{IEEEtran}
\usepackage{amsmath,amsfonts}
\usepackage{algorithm}
\usepackage{algorithmic}
\usepackage{eqparbox}

\newenvironment{proof}{{\indent \indent \it Proof:\quad}}{\hfill $\blacksquare$\par}
\usepackage{array}
\usepackage{textcomp}
\usepackage{stfloats}
\usepackage{url}
\usepackage{cite}
\usepackage{setspace}
\usepackage{amssymb}
\usepackage{float}
\usepackage{subfigure}
\usepackage{verbatim}
\usepackage{graphicx}
\usepackage[colorlinks=true,
            linkcolor=black,
            anchorcolor=black,
            citecolor=black,
            urlcolor=blue
            ]{hyperref}

\newtheorem{definition}{\textbf{Definition}} 

\newtheorem{lemma}{\textbf{Lemma}}

\newtheorem{theorem}{\textbf{Theorem}}

\def\BibTeX{{\rm B\kern-.05em{\sc i\kern-.025em b}\kern-.08em
    T\kern-.1667em\lower.7ex\hbox{E}\kern-.125emX}}
\usepackage{balance}
\begin{document}
\title{Statistical QoS Provisioning Analysis and Performance Optimization in xURLLC-enabled Massive MU-MIMO Networks: A Stochastic Network Calculus Perspective}
\author{Yuang Chen, \emph{Student Member, IEEE}, Hancheng Lu, \emph{Senior Member, IEEE}, Langtian Qin, Chenwu Zhang, and Chang Wen Chen, \emph{Fellow, IEEE}
\thanks{\setlength{\baselineskip}{2\baselineskip}This work was supported by Hong Kong Research Grants Council (GRF-15213322) and National Science Foundation of China (No. U21A20452, No. U19B2044). Yuang Chen, Hancheng Lu, Langtian Qin, and Chenwu Zhang are with the CAS Key Laboratory of Wireless-Optical Communications, School of Information Science and Technology, University of Science and Technology of China, Hefei 230027, China (email: yuangchen21@mail.ustc.edu.cn; hclu@ustc.edu.cn; qlt315@mail.ustc.edu.cn; cwzhang@mail.ustc.edu.cn). Chang Wen Chen is with the Department of Computing, The Hong Kong Polytechnic University, Hong Kong (e-mail: changwen.chen@polyu.edu.hk).}}
\maketitle
\vspace{-3em}
\begin{abstract}
\par In this paper, fundamentals and performance tradeoffs of the neXt-generation ultra-reliable and low-latency communication (xURLLC) are investigated from the perspective of stochastic network calculus (SNC). An xURLLC-enabled massive MU-MIMO system model has been developed to accommodate xURLLC features. By leveraging and promoting SNC, we provide a quantitative statistical quality of service (QoS) provisioning analysis and derive the closed-form expression of upper-bounded statistical delay violation probability (UB-SDVP). Based on the proposed theoretical framework, we formulate the UB-SDVP minimization problem that is first degenerated into a one-dimensional integer-search problem by deriving the minimum error probability (EP) detector, and then efficiently solved by the integer-form Golden-Section search algorithm. Moreover, two novel concepts, EP-based effective capacity (EP-EC) and EP-based energy efficiency (EP-EE) have been defined to characterize the tail distributions and performance tradeoffs for xURLLC. Subsequently, we formulate the EP-EC and EP-EE maximization problems and prove that the EP-EC maximization problem is equivalent to the UB-SDVP minimization problem, while the EP-EE maximization problem is solved with a low-complexity outer-descent inner-search collaborative algorithm. Extensive simulations demonstrate that the proposed framework in reducing computational complexity compared to reference schemes, and in providing various tradeoffs and optimization performance of xURLLC concerning UB-SDVP, EP, EP-EC, and EP-EE.
\end{abstract}


\vspace{-0.4cm}
\section{Introduction}
\IEEEPARstart{A}{s} an emerging and dominant mission-critical and time-sensitive service class of fifth-generation mobile wireless networks and beyond (5G/B5G), ultra-reliable and low-latency communication (URLLC) has sparked enthusiasm in academia and industry \cite{bennis2018ultrareliable,she2021tutorial}. While studies on 5G URLLC are in full swing, the emergence of some novel applications with more stringent quality-of-service (QoS) requirements have prompted the need for neXt-generation URLLC (xURLLC) to become the central vision of sixth-generation (6G) communication systems\cite{she2021tutorial,park2022extreme}. xURLLC is envisioned to enable multifarious innovative services, such as unmanned vehicles, industrial automation, telesurgery, remote training, tactile internet, etc\cite{she2021tutorial,park2022extreme}. All these xURLLC services demand more stringent QoS guarantees, including millisecond-level latency and 99.99999$\%$ packet reliability\cite{park2022extreme}. However, research on xURLLC is still limited and immature.

\par The core technology roadmap to fulfill the low-latency requirements of xURLLC is to organize large amounts of short-packet data communications in highly time-varying wireless networks\cite{park2022extreme, bennis2018ultrareliable, she2021tutorial}. Under these circumstances, standard Shannon's capacity is inappropriate since it is only applicable to the infinite or long blocklength regime. To overcome this theoretical issue, finite blocklength coding (FBC) has been proposed, which reveals the approximate maximum achievable data rate over Additive White Gaussian Noise (AWGN) channels, considering the non-vanishing decoding error probability\cite{polyanskiy2010channel, polyanskiy2011feedback0}. Additionally, Wei \emph{et al.} extended the results in \cite{polyanskiy2010channel,polyanskiy2011feedback0} and investigated the maximum achievable data rates for FBC-based wireless quasi-static fading channels in \cite{yang2014quasi}. Nevertheless, the deterministic QoS provisioning mechanism is far from being able to support the ultra-reliability of the explosively growing xURLLC services in view of the highly time-varying characteristic of wireless fading channels. As a result, the statistical QoS provisioning mechanisms represented by effective capacity and large deviation theory have been actively studied\cite{zhang2020statistical0,zhang2021aoi0,li2021ultra,guo2019resource, amjad2019effective, hou2018burstiness}. A statistical QoS provisioning scheme for massive URLLC (mURLLC) over the cell-free (CF) multiple-input massive multiple-output (MIMO) mobile wireless networks was investigated in \cite{zhang2020statistical0}.
To further alleviate the QoS provisioning issues of mURLLC, a statistical QoS provisioning scheme based on the age-of-information (AoI) concept was developed for mURLLC-enabled CF-MIMO over the unmanned aerial vehicle (UAV) mobile wireless networks\cite{zhang2021aoi0}. Furthermore, the reliability-latency tradeoff and performance bounds of wireless URLLC systems with burst traffics have also been analyzed by exploiting the statistical QoS provisioning mechanism \cite{li2021ultra}.

\par In addition, massive multi-user MIMO (MU-MIMO) may play a potential role in combating the extreme features and enhancing the reliability of xURLLC\cite{park2022extreme,she2021tutorial,zhang2022imreconet}, since it is able to support multiple mobile user equipments (UEs) simultaneously without consuming additional frequency and time resources due to the leveraging of spatial diversity and the deployment of a large number of antennas at the base station (BS)\cite{ngo2013energy,popovski2018wireless,ren2020joint,zeng2019enabling,ostman2021urllc}. This technology has also been utilized in some studies concerning 5G URLLC. For instance, to support URLLC services in the finite blocklength regime, a resource allocation problem over the massive MU-MIMO systems was reported in \cite{ren2020joint}. Additionally, a decoding error probability (EP) minimization problem over the massive MU-MIMO systems with perfect CSI has been proposed in \cite{zeng2019enabling} in order to further enhance the reliability of URLLC. Taking into account imperfect CSI, Johan \"{O}stman \emph{et al.} in \cite{ostman2021urllc} have exploited the spatially correlated channels and pilot contamination to characterize and evaluate the EP for both uplink and downlink of massive MU-MIMO in the finite blocklength regime.

\par Although the aforementioned studies have provided many useful insights for xURLLC, the fundamentals and performance tradeoffs of xURLLC remain elusive. Firstly, the relationships among EP, end-to-end delay, and achievable data rate of xURLLC in the finite blocklength regime are fundamentally different \cite{durisi2016toward,popovski2018wireless,park2022extreme,she2021tutorial}. In particular, the ultra-reliability and low-latency properties of xURLLC are difficult to guarantee simultaneously since the EP for short-packet data communication is always non-vanishing with FBC. However, the QoS provisioning mechanisms for xURLLC short-packet data communications have not been thoroughly investigated. Secondly, the majority of existing URLLC research has focused on average metrics such as average delay, average power consumption, and average throughput \cite{shi2022risk,liu2017latency,salehi2021cooperative}. However, the design overarching core of xURLLC demands more mandate on the tail distributions of reliability and delay \cite{bennis2018ultrareliable,durisi2016toward,she2021tutorial,park2022extreme}, since the tail behavior of xURLLC is essentially related to the tail of stochastic traffic requirement, the tail of delay distribution, interference of intra-cell, mobile UEs power limitation, random location of mobile UEs and shadow fading. Finally, most existing studies have not considered the additional requirements of xURLLC, including high throughput, energy efficiency, extreme delay, and deep shadow fading\cite{popovski2018wireless,park2022extreme,she2021tutorial}, and the performance tradeoffs among these intractable features of xURLLC have not been fully investigated.

\vspace{-0.05cm}
\par Fortunately, stochastic network calculus (SNC) has been proven as an effective mathematical tool in characterizing tail distributions \cite{bennis2018ultrareliable,durisi2016toward,she2021tutorial,jiang2008stochastic,fidler2010survey,fidler2014guide,lubben2014stochastic}, providing us with a crucial methodology for the potential investigation of fundamentals and performance tradeoffs of xURLLC. In this paper, we apply the promoted SNC theory to investigate the performance of the xURLLC-enabled massive MU-MIMO networks, and carry out penetrative statistical QoS provisioning analysis for xURLLC. To characterize the tail distributions and performance tradeoffs of xURLLC, two novel concepts, i.e., EP-based effective capacity (EP-EC) and EP-based energy efficiency (EP-EE), are defined. Based on the proposed theoretical framework, performance optimization of xURLLC is performed. The particularized contributions of this paper are summarized as follows:
\begin{itemize}
  \item An xURLLC-enabled massive MU-MIMO system model with imperfect CSI is developed in the finite blocklength regime, which significantly enhances reliability and thoroughly caters to the features of xURLLC. Subsequently, the SNC theory is utilized and promoted to provide penetrative statistical QoS provisioning analysis for xURLLC. Notably, a novel and succinct operator named \emph{min-deconvolution} is proposed to characterize upper-bounded statistical delay violation probability (UB-SDVP). Moreover, the arrival process and service process of xURLLC-enabled massive MU-MIMO networks are derived, respectively. Lastly, the closed-form expression of UB-SDVP is deduced.
  \item Based on the proposed theoretical framework, we formulate the UB-SDVP minimization problem; it is first degenerated into a one-dimensional integer-search problem by deriving the expression of the minimum EP detector, which is then efficiently solved by an integer-form golden-section search algorithm (IFGSS).
  \item To investigate the performance tradeoffs, the EP-EC and EP-EE maximization problems are subsequently formulated. The EP-EC maximization problem is demonstrated to be equivalent to the UB-SDVP minimization problem for a given transmit power, while the EP-EE maximization problem is solved by proposing a low-complexity outer-descent inner-search collaborative algorithm (ODISC).
  \item Extensive simulations validate and demonstrate that the proposed algorithms can considerably reduce computation complexity compared with reference schemes, revealing various tradeoffs and performance optimization of xURLLC. To further explore these topics, tradeoffs of EP, UB-SDVP, EP-EC, and EP-EE are qualitatively analyzed, and the performance optimization of EP-EC and EP-EE is also qualitatively investigated.
\end{itemize}

\vspace{-0.1cm}
\par The remainder of this paper is organized as follows. In Sec. II, we present an xURLLC-enabled massive MU-MIMO system model in the finite blocklength regime. In Sec. III, we discuss the moment generating function-based SNC (MGF-SNC) theory and the statistical QoS provisioning analysis scheme. In Sec. IV, we display the formulation and the solution to the UB-SDVP minimization problem. In Sec. V, we illustrate the maximization problems of EP-EC and EP-EE. Extensive performance evaluations and analysis are presented in Sec. VI. Conclusions are given in Sec. VII.

\vspace{-1em}
\begin{figure}[htbp]
\centering
\includegraphics[scale=0.60]{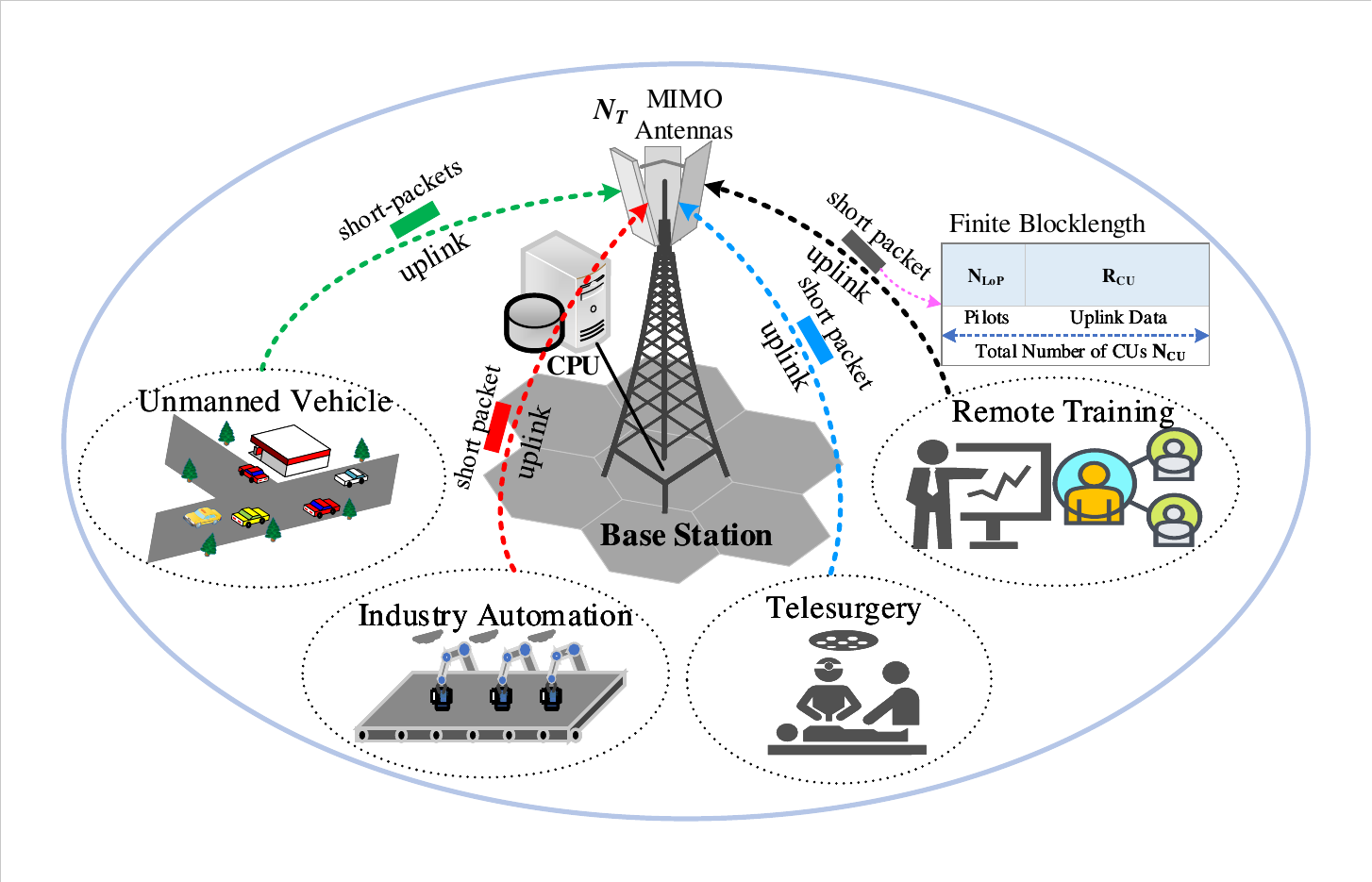}
\caption{The system model for the proposed uplink xURLLC-enabled massive MU-MIMO networks in the finite blocklength regime.}
\label{fig:label}
\end{figure}

\vspace{-2.5em}
\section{The SYSTEM MODEL OF xURLLC-Enabled Massive MU-MIMO Networks}
\par As illustrated in Fig. 1, we propose an uplink xURLLC-enabled massive MU-MIMO wireless networks communication system in the finite blocklength regime, where a BS equipped with $N_{T}$ antennas serves $M$ mobile UEs equipped with a single antenna simultaneously ($N_{T} \gg M$). Let set $\mathcal{M} \triangleq \left\{1,2,\cdots,M\right\}$ denote the index of mobile UEs. Assume each mobile UE $m \in \mathcal{M}$ is randomly located around the BS and entitled to maintain xURLLC short-packet dada communications with the BS through uplink to deliver emergency messages (such as vehicle status, traffic situation, operation query, feedback message, measured data, vital signs, etc.) to the central controller in time. Assume that each mobile UE $m\in \mathcal{M}$ transmits xURLLC short-packet data through the total number of channel uses (CUs) $N_{CU}$, which span across a bandwidth of $B$ MHz and a duration of $t_{DE}$ milliseconds. Thus, the total number of CUs can be denoted as $N_{CU} = B \cdot t_{DE}$, which also represents the code blocklength of one xURLLC short-packet data communication \cite{tse2005fundamentals}. The total number of CUs consists of two parts: 1) $R_{CU}$ CUs for the xURLLC short-packet data transmission of these $M$ mobile UEs; 2) $N_{LoP}$ CUs for estimating the instantaneous fading channel coefficients, resulting in $N_{CU}=N_{LoP} + R_{CU}$. For the sake of simplicity, we mainly focus on the uplink communications; however, the proposed system can also be similarly derived and applied to xURLLC downlink scenarios.
\vspace{-0.6cm}
\subsection{xURLLC Short-Packet Data Communications Model}
\par Let $\boldsymbol{\mathrm{S}}^{(d)} = \big[\boldsymbol{\mathrm{s}}_{1}^{(d)},\cdots,\boldsymbol{\mathrm{s}}_{M}^{(d)}\big]^{T} \in \mathbb{C}^{M \times R_{CU}}$ denote the xURLLC short-packet data of these $M$ mobile UEs, where $\boldsymbol{\mathrm{s}}_{m}^{(d)} = \big[s_{m,1}^{(d)},\cdots,s_{m,R_{CU}}^{(d)}\big]^{T} \in \mathbb{C}^{R_{CU} \times 1}$ is the zero mean and unit variance Gaussian xURLLC short-packet data of the mobile UE $m\in \mathcal{M}$, namely, $\mathbb{E}\big(\big|s_{m,l}^{(d)}\big|^{2}\big)=1$, where $1 \leq m \leq M$ and $1 \leq l \leq R_{CU}$, $\mathbb{E}\big(s_{m,l}^{(d)}s_{i,j}^{(d)}\big) = 0$, $m\neq i$ or $l \neq j$. The uplink channel coefficients matrix from these $M$ mobile UEs to the BS is represented as $\boldsymbol{\mathrm{H}} = \left[\boldsymbol{\mathrm{h}}_{1},\cdots,\boldsymbol{\mathrm{h}}_{M}\right] \in \mathbb{C}^{N_{T} \times M}$. Then, the received signals of the xURLLC short-packet data at the BS can be given as
\vspace{-0.1cm}
\begin{equation}\label{e1}
  \boldsymbol{\mathrm{Y}}^{(d)} = \sqrt{\rho}\sum_{m \in \mathcal{M}} \boldsymbol{\mathrm{h}}_{m}\mathrm{\boldsymbol{\mathrm{s}}}^{(d)}_{m} + \boldsymbol{\mathrm{N}}^{(d)},
\end{equation}
where $\rho$ denotes the transmit power, and $\boldsymbol{\mathrm{N}}^{(d)} \in \mathbb{C}^{N_{T} \times R_{CU}}$ is the AWGN matrix during the xURLLC short-packet data transmission, which has independent identically distributed elements, namely, $\mathcal{CN}\left(0,1\right)$.

\vspace{-0.2cm}
\subsection{Uplink Pilots Training and Channel Estimation in xURLLC-Enabled Massive MU-MIMO Networks}

\par Massive MU-MIMO has been widely regarded as a facilitator of 5G/B5G mobile wireless networks \cite{popovski2018wireless,zhang2022imreconet}, leveraging spatial diversity and enabling channel hardening. As such, massive MU-MIMO has been promised to be exploited to enhance the reliability of xURLLC services and make the xURLLC systems less vulnerable to the fast-fading effects \cite{ren2020joint,ostman2021urllc}. Nevertheless, in order to enjoy the benefits of massive MU-MIMO, CSI must be available at the BS, meaning that channel estimation is essential and indispensable.

\par In this paper, we select the least square channel estimation (LS) for estimating the CSI due to its low complexity and extensive application \cite{bai2019deep,cheng2016optimal}. Worthy of note is that the main thrust of this study remains on the statistical QoS provisioning analysis and performance optimization in xURLLC-enabled massive MU-MIMO wireless networks, even though the channel estimation is taken into consideration. During the uplink pilot training phase, we assume that all mobile UEs synchronously transmit mutually orthogonal pilot sequences to the BS, where the pilot length satisfies $N_{LoP} \geq M$ \cite{ngo2013energy,cheng2016optimal}. Let us denote $\boldsymbol{\mathrm{S}}^{(p)} = \big[\boldsymbol{\mathrm{s}}_{1}^{(p)},\cdots,\boldsymbol{\mathrm{s}}_{M}^{(p)}\big]^{T} \in \mathbb{C}^{M \times N_{LoP}}$ as the pilots of these $M$ mobile UEs, where $\boldsymbol{\mathrm{s}}_{m}^{(p)} = \big[s_{m,1}^{(p)},\cdots,s_{m,N_{LoP}}^{(p)}\big]^{T} \in \mathbb{C}^{N_{LoP} \times 1}$ is the pilots of mobile UE $m\in \mathcal{M}$, which satisfies $\big(\boldsymbol{\mathrm{s}}_{m}^{(p)}\big)^{H} \cdot \boldsymbol{\mathrm{s}}_{m}^{(p)} = 1$ and $\big(\boldsymbol{\mathrm{s}}_{m}^{(p)}\big)^{H} \boldsymbol{\mathrm{s}}_{n}^{(p)} = 0, m \neq n$.

\par According to the preceding discussion, the channel conditions of these $M$ mobile UEs can be estimated at the BS based on the received pilot signals, which can be expressed as follows:
\begin{equation}\label{e2}
  \boldsymbol{\mathrm{Y}}^{(p)} = \sqrt{\rho}\sum_{m \in \mathcal{M}} \boldsymbol{\mathrm{h}}_{m}\mathrm{\boldsymbol{\mathrm{s}}}^{(p)}_{m} + \boldsymbol{\mathrm{N}}^{(p)},
\end{equation}
where $\boldsymbol{\mathrm{N}}^{(p)} \in \mathbb{C}^{N_{T} \times N_{CU}}$ denotes the additive Gaussian noise matrix during the uplink pilots training phase, whose elements are also independent identically distributed, namely, $\mathcal{CN}\left(0,1\right)$. And the channel coefficients matrix $\boldsymbol{\mathrm{H}}$ can be decomposed as
\vspace{-1em}
\begin{equation}\label{e3}
  \boldsymbol{\mathrm{H}} = \boldsymbol{\mathrm{H}}^{\prime} \boldsymbol{\mathrm{\Lambda}}^{1/2}\boldsymbol{\mathrm{\beta}}^{1/2},
\end{equation}
\vspace{-0.5em}
where $\boldsymbol{\mathrm{H}}^{\prime} = \left[\boldsymbol{\mathrm{h}}^{\prime}_{1},\cdots,\boldsymbol{\mathrm{h}}^{\prime}_{M}\right]$, $\boldsymbol{\mathrm{\beta}} = \rm{diag}\left(\beta_{1},\cdots,\beta_{M}\right)$, $\boldsymbol{\mathrm{\Lambda}} = \rm{diag}\left(\lambda_{1},\cdots,\lambda_{M}\right)$. Let vector $\boldsymbol{\mathrm{h}}_{m}^{\prime} = \left[h^{\prime}_{m,1},\cdots,h^{\prime}_{m,N_{T}}\right]^{T}$ represent the small-scale fading channel coefficients from mobile UE $m$ to the BS, which is modeled as a Rayleigh fading distribution with zero mean and unit variance, i.e., $\boldsymbol{\mathrm{h}}_{m}^{\prime}\thicksim \mathcal{CN}(0,\boldsymbol{I}_{N_{T}})$. Additionally, let $\beta_{m}$ denote the shadow fading of mobile UE $m$, which is modeled as a lognormal distribution. Note that shadow fading can not be neglected since it is one of the primary causes of deep fading of the xURLLC services\cite{park2022extreme,bennis2018ultrareliable,she2021tutorial,zeng2019enabling}. Assuming the standard deviation of shadow fading is $\sigma_{\beta}$, we have $10\log_{10}(\beta_{m})\thicksim \mathcal{CN}\big(0,\sigma_{\beta}^{2}\big)$. Furthermore, $\lambda_{m}$ denotes the path loss from mobile UE $m$ to the BS, which is given as $\lambda_{m} = \mu_{cp}\big(d_{m}/d_{min}\big)^{-\alpha_{0}}$ $\left(d_{min} \leq d_{m} \leq d_{max}\right)$, where $d_{m}$ denotes the distance from the mobile UE $m$ to the BS, $\mu_{cp}$ denotes the constant path loss when mobile UEs are at the minimum distance $d_{min}$, $\alpha_{0}$ denotes the path loss factor, $d_{min}$ and $d_{max}$ denote the minimum and maximum distance from the mobile UE $m$ to BS, respectively.

\par According to (\ref{e3}), the uplink channel coefficients vector $\boldsymbol{\mathrm{h}}_{m}$ from mobile UE $m$ to the BS can be rewritten as $\boldsymbol{\mathrm{h}}_{m} = \sqrt{\beta_{m}\lambda_{m}} \boldsymbol{\mathrm{h}}_{m}^{\prime}$. The parameter $\sqrt{\lambda_{m} \beta_{m}}$ models the large-scale path loss and shadow fading, which can be assumed to be independent and constant over many coherence time intervals, and are known a priori\cite{ngo2013energy,tse2005fundamentals}, given that the distances from mobile UEs to BS are usually much larger than the distances between antennas, and both the large-scale path loss and shadow fading are slowly time-varying. In this case, we only focus on estimating the channel coefficients of small-scale fading. By using LS, the estimation of small-scale fading channel coefficients matrix $\boldsymbol{\mathrm{H}}^{\prime}$ can be obtained by multiplying the term $\frac{1}{\sqrt{\rho}N_{LoP}}\big(\boldsymbol{\mathrm{S}}^{(p)}\big)^{H}$ to $\boldsymbol{\mathrm{Y}}^{(p)}$ at the BS, as follows:

\vspace{-1em}
\begin{equation}\label{e4}
  \hat{\boldsymbol{\mathrm{H}}}^{\prime} \!=\! \big[\hat{\boldsymbol{\mathrm{h}}}_{1}^{\prime},\cdots,\hat{\boldsymbol{\mathrm{h}}}_{M}^{\prime}\big] \! = \! \frac{1}{N_{LoP}\sqrt{\rho}} \boldsymbol{\mathrm{Y}}^{(p)}\big(\boldsymbol{\mathrm{S}}^{(p)}\big)^{H}\!\left(\boldsymbol{\beta}\boldsymbol{\Lambda}\right)^{-1/2}.
\end{equation}
\vspace{-1em}
\par According to (\ref{e3}) and (\ref{e4}), the estimation deviation of the small-scale fading channel coefficients vector $\boldsymbol{h}_{m}^{\prime}$ can be derived as follows:
\vspace{-0.2cm}
\begin{equation}\label{e5}
  \hat{\boldsymbol{\mathrm{h}}}_{m}^{\prime}-\boldsymbol{\mathrm{h}}_{m}^{\prime} = \frac{1}{N_{LoP}\sqrt{\rho\lambda_{k}\beta_{k}}} \boldsymbol{\mathrm{N}}^{(p)} \bigg[\big(\boldsymbol{\mathrm{s}}_{m}^{(p)}\big)^{T}\bigg]^{H}.
\end{equation}
\vspace{-1.5em}
\par Correspondingly, the estimation of channel coefficient matrix $\boldsymbol{\mathrm{H}}$ is denoted as
\begin{equation}\label{e6}
  \hat{\boldsymbol{\mathrm{H}}}=\left[\hat{\boldsymbol{\mathrm{h}}}_{1},\cdots,\hat{\boldsymbol{\mathrm{h}}}_{M}\right] = \hat{\boldsymbol{\mathrm{H}}}^{\prime}\boldsymbol{\Lambda}^{1/2}\boldsymbol{\beta}^{1/2}.
\end{equation}
\vspace{-2em}
\par As stated in \cite{fodor2014performance,fodor2015minimizing0}, the expression of the small-scale fading channel coefficients realization $\boldsymbol{\mathrm{h}}_{m}^{\prime}$ conditioned on the estimated matrix $\hat{\boldsymbol{\mathrm{H}}}^{\prime}$ is given as follows:
\vspace{-1em}
\begin{equation}\label{e7}
  \big(\boldsymbol{\mathrm{h}}_{m}^{\prime}|\hat{\boldsymbol{\mathrm{H}}}^{\prime}\big) = \delta_{m}\hat{\boldsymbol{\mathrm{h}}}_{m}^{\prime} +  \zeta_{m} \boldsymbol{\mathrm{e}}_{m},
\end{equation}
where $\delta_{m} = \frac{\rho N_{LoP} \lambda_{m} \beta_{m}}{\rho N_{LoP} \lambda_{m} \beta_{m} + 1}$, $\zeta_{m} = \frac{1}{\sqrt{\lambda_{m}\beta_{m}}}$, and $\boldsymbol{\mathrm{e}}_{m} \thicksim \mathcal{CN}\big(0,\frac{\delta_{m}}{\rho N_{LoP}} \boldsymbol{\mathrm{I}}_{N_{T}}\big)$ denotes the estimation errors of the small-scale channel coefficients by using LS detector.
\par Substituting (\ref{e3}) into (\ref{e7}), we can obtain that
\vspace{-0.1cm}
\begin{equation}\label{e8}
   \left(\boldsymbol{\mathrm{H}}|\hat{\boldsymbol{\mathrm{H}}}^{\prime}\right) = \boldsymbol{\mathrm{D}} + \boldsymbol{\mathrm{E}},
\end{equation}
where $\boldsymbol{\mathrm{D}} = \left[\boldsymbol{\mathrm{d}}_{1},\cdots,\boldsymbol{\mathrm{d}}_{M}\right]$, and $\boldsymbol{\mathrm{d}}_{m} = \delta_{m}\sqrt{\lambda_{m}\beta_{m}} \hat{\boldsymbol{\mathrm{h}}}_{m}^{\prime}$, $\boldsymbol{\mathrm{E}} = \left[\zeta_{1}\boldsymbol{\mathrm{e}}_{1},\cdots,\zeta_{M}\boldsymbol{\mathrm{e}}_{M}\right]$. Then, the channel coefficient vector $\boldsymbol{\mathrm{h}}_{m}$ conditioned on the estimated matrix $\hat{\boldsymbol{\mathrm{H}}}^{\prime}$ is given as
\vspace{-0.1cm}
\begin{equation}\label{e9}
  \big(\boldsymbol{\mathrm{h}}_{m}|\hat{\boldsymbol{\mathrm{H}}}^{\prime}\big) = \delta_{m}\sqrt{\lambda_{m}\beta_{m}} \hat{\boldsymbol{\mathrm{h}}}_{m}^{\prime} +  \zeta_{m} \boldsymbol{\mathrm{e}}_{m}.
\end{equation}

\vspace{-0.5cm}
\subsection{Linear Detection for the Received xURLLC Short-Packet Data}
\vspace{-0.2cm}
\par Combining (\ref{e1}) and (\ref{e7})$-$(\ref{e9}), the xURLLC short-packet data $\boldsymbol{\mathrm{Y}}^{(d)} = \left[\boldsymbol{\mathrm{y}}_{1}^{(d)},\cdots,\boldsymbol{\mathrm{y}}_{R_{CU}}^{(d)}\right] \in \mathbb{C}^{L\times R_{CU}}$ of the $M$ mobile UEs received at the BS can be expressed as
\vspace{-0.2cm}
\begin{equation}\label{e10}
 \begin{aligned}
    \boldsymbol{\mathrm{Y}}^{(d)} \!& = \! \sqrt{\rho}\big(\boldsymbol{\mathrm{H}}|\hat{\boldsymbol{\mathrm{H}}}^{'}\!\big)\boldsymbol{\mathrm{S}}^{(d)}\!+\!\boldsymbol{\mathrm{N}}^{(d)} \!=\! \sqrt{\rho}\left(\boldsymbol{\mathrm{D}} + \boldsymbol{\mathrm{E}}\right)\boldsymbol{\mathrm{S}}^{(d)}\!+\boldsymbol{\mathrm{N}}^{(d)} \\
     & = \sqrt{\rho}\boldsymbol{\mathrm{D}}\boldsymbol{\mathrm{S}}^{(d)} + \sqrt{\rho}\boldsymbol{\mathrm{E}}\boldsymbol{\mathrm{S}}^{(d)}+\boldsymbol{\mathrm{N}}^{(d)}.
 \end{aligned}
\end{equation}
\par More specifically, the detection xURLLC short-packet data $\boldsymbol{\mathrm{y}}_{m, r_{CU}}$ of the mobile UE $m$ ($r_{CU} \in \left\{1,2,\cdots,R_{CU}\right\}$) can be given as follows:
\vspace{-0.2cm}
\begin{equation}\label{e11}
  \begin{aligned}
       \boldsymbol{\mathrm{y}}_{m, r_{CU}} &= \boxed{\underbrace{\sqrt{\rho}s_{m,r_{CU}}^{(d) }\boldsymbol{\mathrm{d}}_{m}}_{\text{Mobile UE-$m$}}} + \underbrace{\sum\limits_{i = 1, i \neq m}^{M} \sqrt{\rho}s_{i,r_{CU}}^{(d)}\boldsymbol{\mathrm{d}}_{i}}_{\text{Interference from other mobile UEs}} + \underbrace{\sum\limits_{i = 1}^{M} \sqrt{\rho}\zeta_{i}s_{i,r_{CU}}^{(d)}\boldsymbol{\mathrm{e}}_{i}}_{\text{Estimation errors}} + \underbrace{\boldsymbol{\mathrm{n}}_{r_{CU}}^{(d)}}_{\text{Noise}}.
  \end{aligned}
\end{equation}
\vspace{-2em}
\par Note that the estimated channel coefficients matrix $\hat{\boldsymbol{\mathrm{H}}}$ is used as the true channel coefficients at the BS.
The first item in (\ref{e11}) represents the received xURLLC short-packet data of mobile UE $m$, while the remaining three items are deemed as interference or noise. By utilizing linear detection\cite{studer2011asic,ren2020joint}, the received URLLC short-packet data $\boldsymbol{\mathrm{y}}_{m, r_{CU}}$ for mobile UE $m$ can be processed as
\vspace{-1.5em}
\begin{equation}\label{e12}
  \boldsymbol{\mathrm{y}}_{m, r_{CU}}^{D} = \boldsymbol{\mathrm{L}}_{m}\boldsymbol{\mathrm{y}}_{m, r_{CU}},
\end{equation}
\vspace{-0.5em}
where $\boldsymbol{\mathrm{L}}_{m}$ denotes the linear detector matrix of mobile UE $m$.

\par Given the pilot length $N_{LoP}$ and linear detector matrix $\boldsymbol{\mathrm{L}}_{m}$ \cite{ngo2013energy,ren2020joint}, the SINR for each mobile UE $m$ can be given as (\ref{e13}), where $\boldsymbol{\mathrm{G}} \triangleq \boldsymbol{\mathrm{D}}\boldsymbol{\mathrm{D}}^{H} + \frac{1}{\omega}\boldsymbol{\mathrm{I}}_{N_{T}}=\hat{\boldsymbol{\mathrm{H}}}\boldsymbol{\mathrm{\delta}}\boldsymbol{\mathrm{\delta}}\hat{\boldsymbol{\mathrm{H}}}^{H}+\frac{1}{\omega}\boldsymbol{\mathrm{I}}_{N_{T}}$, and $\omega = \left(\sum_{i=1}^{M}\frac{\lambda_{i}\beta_{i}}{\rho N_{LoP} \lambda_{i}\beta_{i} + 1} + \frac{1}{\rho} \right)^{-1}$.
In this paper, the probability distribution function (PDF) of the SINR $\hat{\gamma}_{m}\left(N_{LoP},\boldsymbol{\mathrm{L}}_{m}\right)$ for each mobile UE $m$ can be approximated by a Gamma distribution as follows\cite{li2005distribution, ngo2013energy}
\begin{figure*}[t]
\centering
\hrulefill
\begin{equation}\label{e13}
\setstretch{0.95}
  \begin{aligned}
      \hat{\gamma}_{m}\left(N_{LoP},\boldsymbol{\mathrm{L}}_{m}\right) &= \frac{\big \|\sqrt{\rho}s_{m,r_{CU}}^{(d)}\boldsymbol{\mathrm{L}}_{m}\boldsymbol{\mathrm{d}}_{m}\big \|^{2}}{\mathbb{E}\left\{ \bigg \|\boldsymbol{\mathrm{L}}_{m}\left(\sqrt{\rho}\left(\sum\limits_{i = 1, i \neq m}^{M} s_{i,r_{CU}}^{(d)}\boldsymbol{\mathrm{d}}_{i} +\sum\limits_{i = 1}^{M} \zeta_{i}s_{i,r_{CU}}^{(d)}\boldsymbol{\mathrm{e}}_{i}\right) +\boldsymbol{\mathrm{n}}_{r_{CU}}^{(d)}\right)\bigg \|^{2}\right\}}\\
      & = \frac{\rho \big \|\boldsymbol{\mathrm{L}}_{m}\boldsymbol{\mathrm{d}}_{m}\big \|^{2}}{\rho \boldsymbol{\mathrm{L}}_{m} \left(\sum\limits_{i=1,i \neq m}^{M} \boldsymbol{\mathrm{d}}_{i}\boldsymbol{\mathrm{d}}_{i}^{H} + \frac{\rho}{\omega}\boldsymbol{\mathrm{I}}_{N_{T}}\right)\boldsymbol{\mathrm{L}}_{m}^{H}} = \frac{\boldsymbol{\mathrm{L}}_{m}\boldsymbol{\mathrm{d}}_{m}\boldsymbol{\mathrm{d}}_{m}^{H}\boldsymbol{\mathrm{L}}_{m}^{H}}{\boldsymbol{\mathrm{L}}_{m} \left(\boldsymbol{\mathrm{G}} - \boldsymbol{\mathrm{d}}_{m}\boldsymbol{\mathrm{d}}_{m}^{H}\right)\boldsymbol{\mathrm{L}}_{m}^{H}}.
  \end{aligned}
\end{equation}
\hrulefill
\end{figure*}
\vspace{-0.1cm}
\begin{equation}\label{e14}
  p_{\hat{\gamma}_{m}\left(N_{LoP},\boldsymbol{\mathrm{L}}_{m}\right)}\left(x\right) = \frac{x^{\mu_{m}-1}e^{-x/\nu_{m}}}{\Gamma\left(\mu_{m}\right)\nu_{m}^{\mu_{m}}},
\end{equation}
\vspace{-0.5em}
where $\Gamma(x) = \int_{0}^{\infty}t^{x-1}e^{-t}\mathrm{d}t$, $\hat{\lambda}_{m} = \lambda_{m}\beta_{m}\delta_{m}$, and
\begin{subequations}\label{e15}
\setstretch{0.95}
   \begin{align}
      & \mu_{m} = \frac{\left(N_{T}-M+1+(M-1)\psi_{m}\right)^{2}}{N_{T}-M+1+(M-1)\kappa_{m}},\\
      & \nu_{m} = \frac{N_{T}-M+1+(M-1)\kappa_{m}}{N_{T}-M+1+(M-1)\psi_{m}}\omega\hat {\lambda}_{m},\\
      & \psi_{m} = \frac{1}{M\!-\!1}\!\!\!\sum\limits_{i=1,i \neq m}^{M} \!\frac{1}{N_{T}\omega \hat{\lambda}_{i}\!\big(1\!-\!\frac{M-1}{N_{T}}\!+\!\frac{M-1}{N_{T}}\psi_{m}\big)\!+\!1},\\
      & \kappa_{m} \bigg(1+\sum\limits_{i=1,i \neq m}^{M} \frac{\omega\hat{\lambda}_{i}}{\big(N_{T}\omega \hat{\lambda}_{i}\big(1-\frac{M-1}{N_{T}}+\frac{M-1}{N_{T}}\psi_{m}\big)+1\big)^{2}}\bigg) \nonumber \\
      & = \sum\limits_{i=1,i \neq m}^{M} \frac{\omega\hat{\lambda}_{i}\psi_{m} + \frac{1}{M-1}}{\big(N_{T}\omega \hat{\lambda}_{i}\big(1-\frac{M-1}{N_{T}}+\frac{M-1}{N_{T}}\psi_{m}\big)+1\big)^{2}}.
   \end{align}
\end{subequations}
\vspace{-0.6cm}
\subsection{Channel Model for the Maximum Achievable Data Rate in the Finite Blocklength Regime}
\par Referring to FBC theory\cite{polyanskiy2010channel,polyanskiy2011feedback0,yang2014quasi}, given the instantaneous SINR $\hat{\gamma}_{m}\left(N_{LoP},\boldsymbol{\mathrm{L}}_{m}\right)$ and decoding error probability $\hat{\epsilon}_{m}\left(N_{LoP},\boldsymbol{\mathrm{L}}_{m}\right)$, the maximum achievable data rate (in bpcu) of mobile UE $m$ can be closely approximated by
\vspace{-0.2cm}
\begin{equation}\label{e16}
  r_{m}\left(\hat{\gamma}_{m},\hat{\epsilon}_{m}\right) = C\left(\hat{\gamma}_{m}\right) -\sqrt{\frac{\mathcal{V}\left(\hat{\gamma}_{m}\right)}{R_{CU}}}Q^{-1}\left(\hat{\epsilon}_{m}\right),
\end{equation}
where $Q^{-1}\left(\cdot\right)$ is the inverse of Gaussian function $Q\left(x\right)=\int_{x}^{\infty}\frac{1}{\sqrt{2\pi}}e^{-x^{2}/2} \mathrm{d}x$, $C\left(\hat{\gamma}_{m}\right)$ denotes the classic Shannon Capacity, and $\mathcal{V}\left(\hat{\gamma}_{m}\right)$ denotes the channel dispersion. The expressions of $C\left(\hat{\gamma}_{m}\right)$ and $\mathcal{V}\left(\hat{\gamma}_{m}\right)$ are respectively as follows:
\vspace{-0.3cm}
\begin{subequations}\label{e17}
  \begin{align}
     & C\left(\hat{\gamma}_{m}\right) = \log_{2}\big(1+ \hat{\gamma}_{m}\left(N_{LoP},\boldsymbol{\mathrm{L}}_{m}\right)\big),\\
     & \mathcal{V}\left(\hat{\gamma}_{m}\right) = \big(1 -\big(1+\hat{\gamma}_{m}\left(N_{LoP},\boldsymbol{\mathrm{L}}_{m}\right)\big)^{-2}\big)\left(\log_{2}e\right)^{2},
  \end{align}
\end{subequations}
\vspace{-2em}
\par Combining (\ref{e16}), (\ref{e17}a), and (\ref{e17}b), the maximum achievable data rate $r_{m}\left(\hat{\gamma}_{m},\hat{\epsilon}_{m}\right)$ can be reformulated as
\vspace{-0.5em}
\begin{equation}\label{e18}
  r_{m}\left(\hat{\gamma}_{m},\hat{\epsilon}_{m}\right) = \log_{2}\tilde{f}\left(\hat{\gamma}_{m},\hat{\epsilon}_{m}\right),
\end{equation}
\vspace{-0.3cm}
where
\vspace{-0.3cm}
\begin{equation}\label{e19}
    \!\!\!\tilde{f}\left(\hat{\gamma}_{m},\hat{\epsilon}_{m}\right) \!=\!
   \begin{cases}
    \!\! \frac{1+\hat{\gamma}_{m}}{exp\left\{\sqrt{\frac{\hat{\gamma}_{m}^{2}+2\hat{\gamma}_{m} }{R_{CU}\left(1+\hat{\gamma}_{m}\right)^{2}}}Q^{-1}\left(\hat{\epsilon}_{m}\right) \right\}}, \!\!\!\! & \text{if}\hat{\gamma}_{m} \!>\! \gamma_{0}, \\
    1, \!\!\!\! & \text{if }\hat{\gamma}_{m} \!\leqslant\! \gamma_{0},
   \end{cases}
\end{equation}
and $\gamma_{0}$ is the maximum threshold which corresponds to the value of SINR when $r_{m}\left(\hat{\gamma}_{m},\hat{\epsilon}_{m}\right) = 0$.
\vspace{-2em}
\section{Theoretical Framework: Statistical QoS provisioning Analysis Scheme Based on the Promoted MGF-SNC}
\vspace{-0.1cm}
\par In this section, we propose a theoretical framework for the developed xURLLC-enabled massive MU-MIMO networks by utilizing the promoted MGF-SNC theory. Firstly, the MGF-SNC theory is promoted by introducing a novel and succinct operator named min-deconvolution $\widehat{\oslash}$ to describe UB-SDVP. Secondly, the expressions of the arrival and service processes are derived, respectively. Lastly, the closed-form expression of UB-SDVP is deduced.
\vspace{-1.5em}
\subsection{The Promotion of MGF-SNC}
\vspace{-0.5em}
\par For the convenience of system-level analysis, a statistical QoS-driven stochastic discrete-time queueing system is considered. Assume that short-packet xURLLC data is served in accordance with a first-come-first-serve (FCFS) policy. The cumulative arrival, service, and departure processes during the time interval $\left[s,t\right)$ can be denoted as $A_{m}\left(s,t\right) \triangleq \sum_{i=s}^{t-1}a_{m}\left(i\right)$, $S_{m}\left(s,t\right) \triangleq \sum_{i=s}^{t-1}s_{m}\left(i\right)$, and $D_{m}\left(s,t\right) \triangleq \sum_{i=s}^{t-1}d_{m}\left(i\right)$, respectively, where $a_m\left(i\right)$ denotes arrival data bits generated by mobile UE $m$ at the $i$-th time slot, $d_{m}(i)$ represents the departure data bits that the mobile UE $m$ successfully deliveries to the BS at the $i$-th time slot, and $s_{m}\left(i\right)$ indicates the service rate at which the mobile UE $m$ is successfully served and receives the acknowledged data bits at the $i$-th time slot. In SNC, the convolution operator $\otimes$ and deconvolution operator $\oslash$ are two critical operators used for characterizing the statistical performance of queueing systems, which are usually described by $\left(\mathrm{min},+\right)$-algebras\cite{jiang2008stochastic, fidler2010survey, al2014network, fidler2014guide, lubben2014stochastic, yang2018low}. To make the subsequent derivation and presentation of UB-SDVP more concise, the operator min-deconvolution $\widehat{\oslash}$ is defined in \textrm{Definition 1}.

\begin{definition}
  (min-deconvolution $\widehat{\oslash}$): Assume that the cumulative arrival process $A_{m}\left(s,t\right)$ and service process $S_{m}\left(s,t\right)$ of mobile UE $m$ are independent each other, then the expression of the operator $\widehat{\oslash}$ can be given as follows:
   \begin{equation}\label{e20}
     \begin{aligned}
       & \!\!\!\boldsymbol{\mathrm{M}}\!_{A_{m} \widehat{\oslash} S_{m}}\!\! \left(\theta_{m},\!s,\!t\right)\! =\!\!\! \sum\limits_{u=0}^{\min\left\{\!s,t\!\right\}}\!\! \mathbb{M}_{A_{m}}\!\!\left(\theta_{m},\!u,\!t\right) \cdot \overline{\mathbb{M}}_{S_{m}}\!\!\left(\theta_{m},\!u,\!s\right),
     \end{aligned}
   \end{equation}
   where $\theta_{m} \geq 0$ denotes the QoS exponent of mobile UE $m$, $\mathbb{M}_{A_{m}}$ and $\overline{\mathbb{M}}_{S_{m}}$ refer to the MGF of $A_{m}$ and inverse-MGF of $S_{m}$, respectively\footnote{Given a stochastic process $U\left(s,t\right)$, $0\leq s \leq t$, the MGF of $U\left(s,t\right)$ is $\mathbb{M}_{U}\left(\theta,s,t\right) \triangleq \mathbb{E}\left[e^{\theta U\left(s,t\right)}\right]$, while the inverse-MGF of $U\left(s,t\right)$ is $\overline{\mathbb{M}}_{U}\left(\theta,s,t\right) \triangleq \mathbb{E}\left[e^{-\theta U\left(s,t\right)}\right]$. In SNC, the larger $\theta_{m}$ (e.g., $\theta_{m} \rightarrow \infty$) indicates more stringent statistical QoS requirements, while the smaller $\theta_{m}$ (e.g., $\theta_{m} \rightarrow 0$) implies looser statistical QoS requirements\cite{jiang2008stochastic, fidler2010survey, al2014network, fidler2014guide, lubben2014stochastic,yang2018low,hou2018burstiness,mei2022delay,zhao2020latency,chen2021end,mei2022statistical}.}.
\end{definition}

\par The concept of statistical QoS guarantees has been extensively studied for time-sensitive networks \cite{yang2018low,hou2018burstiness,mei2022delay,zhao2020latency,chen2021end,mei2022statistical}. However, the closed-form expression of the statistical delay violation probability (SDVP) is typically unavailable to derive. Exploiting the definition of min-deconvolution $\widehat{\oslash}$, we formulate the expression of SDVP.

\begin{theorem}
Given the arrival process $A_{m}\left(s,t\right)$ and service process $S_{m}\left(s,t\right)$ of mobile UE $m$, the statistical delay-bound violation probability of mobile UE $m$ can be characterized by the operator min-deconvolution $\widehat{\oslash}$ as follows:
\begin{equation}\label{e21}
 \begin{aligned}
  \mathbb{P}\left(W_{m}\left(t\right) \geq w \right) & \leq \inf_{\theta_{m} \geq 0} \boldsymbol{\mathrm{M}}_{A_{m} \widehat{\oslash} S_{m}} \left(\theta_{m},t+w,t\right).
 \end{aligned}
\end{equation}
\end{theorem}

\begin{proof}
The common form of statistical delay violation probability in \cite{jiang2008stochastic, fidler2010survey,fidler2014guide,lubben2014stochastic} can be denoted as follows:
\begin{equation}\label{e22}
   \mathbb{P}\left(W_{m}\left(t\right) \geq w \right) \leq \mathbb{P}\big( \left(A_{m}\oslash S_{m}\right)\left(t+w,t\right) \geq 0 \big)
\end{equation}
\par According to Chernoff's bound, given a stochastic process $X$, the following inequality holds for any $x > 0$, as follows:
\begin{equation}\label{e23}
  \mathbb{P}\left(X \geq x\right) \leq e^{-\theta x}\mathbb{M}_{X}\left(\theta\right),
\end{equation}
\par Substituting (\ref{e22}) into (\ref{e23}), we can immediately infer that
\begin{equation}\label{e24}
   \begin{aligned}
     \mathbb{P}\left(W_{m}\left(t\right) \geq w \right) & \leq \mathbb{P}\big( \left(A_{m}\oslash S_{m}\right)\left(t+w,t\right) \geq 0 \big)\\
     & \leq \inf_{\theta_{m} > 0} \mathbb{M}_{A_{m} \oslash S_{m}} \left(\theta_{m},t+w,t\right).
   \end{aligned}
\end{equation}

\par According to the definition of MGF, we have
\vspace{-0.2cm}
\begin{equation}\label{e25}
  \begin{aligned}
      & \mathbb{M}_{A_{m} \oslash S_{m}} \!\! \left(\theta_{m},t\!+\!w,t\right)\! \leq \! \sum\limits_{u=0}^{t+w}\!\mathbb{M}_{A_{m}}\!\!\left(\theta_{m},\!u,\!t\right)\!\cdot\! \overline{\mathbb{M}}_{S_{m}}\!\!\left(\theta_{m},u,t\!+\!w\right).
  \end{aligned}
\end{equation}
\par Referring to \emph{\textbf{Definition 1}} and (\ref{e25}), we can obtain that
\vspace{-0.1cm}
\begin{equation}\label{e26}
  \inf_{\theta_{m} > 0} \mathbb{M}_{A_{m} \oslash S_{m}} \left(\theta_{m},t + w,t\right)  \leq  \inf_{\theta_{m} > 0} \boldsymbol{\mathrm{M}}_{A_{m} \widehat{\oslash}S_{m}} \left(\theta_{m},t + w,t\right).
\end{equation}
\par As a result, according to (\ref{e22})-(\ref{e26}), the proof of \textbf{\emph{Theorem 1}} is concluded.
\end{proof}

\vspace{-0.4cm}
\subsection{The Upper-Bound of Statistical Delay Violation Probability}
\par The promotion of MGF-SNC provides a theoretical bedding for deriving the closed-form expression of UB-SDVP. Assume that the increments of the increments of the arrival process $A_{m}\left(s,t\right)$ and service process $S_{m}\left(s,t\right)$ at all different time slots are independent and identically distributed, i.e., $a_{m}(i)$ and $s_{m}(i)$, their respective MGFs and inverse-MGFs can be expressed as follows:
\vspace{-1em}
\begin{subequations}\label{e27}
  \begin{align}
     & \mathbb{M}_{A_{m}} \! \left(\theta_{m},\!s,\!t\right) \! = \! \big(\mathbb{E}\left[e^{\theta_{m}a_{m}}\right]\!\big)^{t-s} \!=\! \big(\mathbb{M}_{a_{m}}\!\!\left(\theta_{m}\right)\!\!\big)^{t-s}, \\
     & \overline{\mathbb{M}}_{S_{m}}\!\left(\theta_{m},\!s,\!t\right) \! = \! \big(\mathbb{E}\left[e^{-\theta_{m}s_{m}}\right]\!\big)^{t-s} \!=\! \big(\overline{\mathbb{M}}_{s_{m}}\!\!\left(\theta_{m}\right)\!\!\big)^{t-s}.
  \end{align}
\end{subequations}
\vspace{-2em}
\par By combining \textrm{Definition 1} and (\ref{e27}), the UB-SDVP can be derived as follows:
\begin{theorem}
    Given a target delay $d_{th}$, the statistical delay violation probability for each mobile UE $m$ is upper bounded by
    \vspace{-0.1cm}
    \begin{equation}\label{e28}
      \mathbb{P}\left(W_{m}\!\left(t\right)\!\geq \!d_{th} \right) \!\leq\! \inf_{\theta_{m} > 0}\! \left\{\! \frac{\big(\overline{\mathbb{M}}_{s_{m}}\left(\theta_{m}\right)\big)^{d_{th}}}{1 \!-\! \mathbb{M}_{a_{m}}\!\left(\theta_{m}\right)\cdot \overline{\mathbb{M}}_{s_{m}}\!\left(\theta_{m}\right)} \!\!\right\}
    \end{equation}
    where $W_{m}\left(t\right)$ denotes the actual delay of mobile UE $m$, $\mathbb{M}_{a_{m}}\left(\theta_{m}\right)\cdot \overline{\mathbb{M}}_{s_{m}}\left(\theta_{m}\right) < 1$ denotes the stability condition.
\end{theorem}

\begin{proof}
  According to \emph{\textbf{Definition 1}}, the min-deconvolution $\widehat{\oslash}$ between $A_{m}(s,t)$ and $S_{m}(s,t)$ is represented by (\ref{e29}),
\begin{figure*}[b]
  \setstretch{0.90}
\centering
\hrulefill
  \begin{equation}\label{e29}
     \begin{aligned}
      & \boldsymbol{\mathrm{M}}_{A_{m}\widehat{\oslash} S_{m}} \left(\theta_{m},s,t\right) \\
      &\overset{(a)}{\leq} \sum_{u=0}^{\min\left(s,t\right)} \left(\mathbb{M}_{a_{m}}\left(\theta_{m}\right)\right)^{t-u} \cdot \left(\overline{\mathbb{M}}_{s_{m}}\left(\theta_{m}\right)\right)^{s-u}
      \overset{(b)}{=} \left(\mathbb{M}_{a_{m}}\left(\theta_{m}\right)\right)^{t-s}\cdot \sum\limits_{v = \tau}^{s}\left(\mathbb{M}_{a_{m}}\left(\theta_{m}\right) \overline{\mathbb{M}}_{s_{m}}\left(\theta_{m}\right)\right)^{v}\\
      & \overset{(c)}{\leq} \left(\mathbb{M}_{a_{m}}\left(\theta_{m}\right)\right)^{t-s}\cdot \sum\limits_{v = \tau}^{\infty}\left(\mathbb{M}_{a_{m}}\left(\theta_{m}\right) \overline{\mathbb{M}}_{s_{m}}\left(\theta_{m}\right)\right)^{v} \overset{(d)}{=} \frac{\left(\mathbb{M}_{a_{m}}\left(\theta_{m}\right)\right)^{t-s}\cdot \left(\mathbb{M}_{a_{m}}\left(\theta_{m}\right) \overline{\mathbb{M}}_{s_{m}}\left(\theta_{m}\right)\right)^{\tau} }{1 - \mathbb{M}_{a_{m}}\left(\theta_{m}\right) \overline{\mathbb{M}}_{s_{m}}\left(\theta_{m}\right)}
     \end{aligned}
  \end{equation}
\hrulefill
\end{figure*}
  where $\tau = \max\left\{0,s-t\right\}$. By substituting (\ref{e27}a) and (\ref{e27}b) into (\ref{e20}), inequality (a) can be obtained. Through a variable replacement, which sets $v = s-u$, equality $(b)$ is derived from inequality $(a)$. Scaling the upper bound of the summation sign in $(b)$ from $s$ to $\infty$, inequality $(c)$ is obtained. Applying the property of the geometric series convergence and assuming that the stability condition $\mathbb{M}_{a_{m}}\left(\theta_{m}\right)\cdot \overline{\mathbb{M}}_{s_{m}}\left(\theta_{m}\right) < 1$ holds, equation $(d)$ can be obtained.
  \par By combining \textbf{\emph{Theorem 1}} and (\ref{e29}), we can finally derive that
  \begin{equation}\label{e30}
     \begin{aligned}
        \mathbb{P}\!\left(W_{m}(t) \!\geq\! d_{th}\right) & \leq\! \inf_{\theta_{m} > 0}\!\!\left\{\!\!\frac{\left(\overline{\mathbb{M}}\!_{s_{m}}\left(\theta_{m}\right)\right)^{d_{th}}}{1-\mathbb{M}\!_{a_{m}}\left(\theta_{m}\right)\overline{\mathbb{M}}\!_{s_{m}}\left(\theta_{m}\right)}\!\!\right\},
     \end{aligned}
  \end{equation}
  \par So the proof of \textbf{\emph{Theorem 2}} is concluded.
\end{proof}

\par Theorem 2 states that the closed-form expressions for the MGF of the arrival increments $a_{m}$, i.e., $\mathbb{M}_{a_{m}}\left(\theta_{m}\right)$, and the inverse-MGF of the service increments $s_{m}$, i.e., $\overline{\mathbb{M}}_{s_{m}}\left(\theta_{m}\right)$, are essential for further determining UB-SDVP. This motivates the follow-up work discussed in this paper.

\subsubsection{The MGF of the Arrival Process}
\par We consider the arrival process of each mobile UE $m$ follows a Poission distribution. Then, the MGFs of $A_{m}(\theta_{m},s,t)$ and $a_{m}(i)$ can be respectively represented as follows:
\vspace{-0.5em}
\begin{subequations}\label{e31}
   \begin{align}
     & \mathbb{M}_{A_{m}}\!\!\left(\theta_{m},\!s,\!t\right) \!=\! \mathbb{E}\bigg[\!\bigg(\prod_{i=s}^{t-1}\!\!e^{a_{m}(i)}\!\!\bigg)^{\!\!\theta_{m}}\!\bigg]\! = \! \left(\mathbb{M}_{a_{m}}\!\!\left(\theta_{m}\right)\right)^{s-t},\\
     & \mathbb{M}_{a_{m}}\left(\theta_{m}\right) = \sum\limits_{i=1}^{\infty}e^{i\theta_{m}}\frac{(\lambda_{m}^{\dag})^{i}e^{-\lambda_{m}^{\dag}}}{i!} = e^{\lambda_{m}^{\dag}\left(e^{\theta_{m}}-1\right)},
   \end{align}
\end{subequations}
\vspace{-0.5em}
where $\lambda_{m}^{\dag}$ denotes the average arrival rate of $A_{m}\left(s,t\right)$.

\subsubsection{The Inverse-MGF of the Service Process}
\par The inverse-MGF of the service process for each mobile UE $m$ is related to the channel conditions $\hat{\gamma}_{m}\left(N_{LoP},\boldsymbol{\mathrm{L}}_{m}\right)$, and the relevant results are given in \textrm{Theorem 3}.

\begin{theorem}
     Given the QoS exponent $\theta_{m}$ and decoding error probability $\hat{\epsilon}_{m}$, the inverse-MGF of the service process for each mobile UE $m$, denoted by $\overline{\mathbb{M}}_{s_{m}}\left(\theta_{m}\right)$, can be given as follows:
     \vspace{-1em}
    \begin{equation}\label{e32}
       \begin{aligned}
         \!\!\!\overline{\mathbb{M}}_{s_{m}}\!\!\left(\theta_{m}\right)\cong \mathcal{H}_{m}\!\left(N_{LoP},\boldsymbol{\mathrm{L}}_{m},\hat{\epsilon}_{m}\right) + \left(1\!-\!\hat{\epsilon}_{m}\right)\overline{\mathbb{M}}_{\hat{\gamma}_{m}}\!\!\left(\Theta_{m}\right), 
       \end{aligned}
    \end{equation}
    \vspace{-0.5em}
    where
    \vspace{-0.5em}
    \begin{equation}\label{e33}
        \mathcal{H}_{m}\left(N_{LoP},\boldsymbol{\mathrm{L}}_{m},\hat{\epsilon}_{m}\right) \triangleq \hat{\epsilon}_{m} + \left(1-\hat{\epsilon}_{m}\right)\int_{0}^{\gamma_{0}}p_{\hat{\gamma}_{m}}(x)\mathrm{d}x,
    \end{equation}
    \vspace{-0.5em}
    and $\Theta_{m} = \frac{\theta_{m}R_{CU}}{\ln2}$ denotes the normalization of QoS exponent $\theta_{m}$.
\end{theorem}
\begin{proof}
     By combining \textbf{\emph{Definition 1}} and (16), the inverse-MGF of service process for each mobile UE $m$ is given as follows:
    \begin{equation}\label{e34} 
       \begin{aligned}          
          & \overline{\mathbb{M}}_{s_{m}}\left(\theta_{m}\right)= \mathbb{E}_{\hat{\gamma}_{m}}\left[e^{-\theta_{m}R_{CU}r\left(N_{LoP},\boldsymbol{\mathrm{L}}_{m}\right)}\right]\\
          & = \!\!\! \int_{0}^{\infty} \!\!\! \bigg(\!\! \left(1\!-\!\hat{\epsilon}_{m}\right)e^{-\theta_{m}R_{CU}\!\log_{2}\!\tilde{f}\left(\hat{\gamma}_{m},\hat{\epsilon}_{m}\right)} \!+\! \hat{\epsilon}_{m} \!\! \bigg)\cdot p_{\hat{\gamma}_{m}}\!\left(x\right)\mathrm{d}x \\
          & = \hat{\epsilon}_{m} + \left(1-\hat{\epsilon}_{m}\right)\int_{0}^{\infty} e^{-\theta_{m}R_{CU}\log_{2}\tilde{f}\left(\hat{\gamma}_{m},\hat{\epsilon}_{m}\right)}\cdot p_{\hat{\gamma}_{m}}\left(x\right) \mathrm{d}x \\
          & = \hat{\epsilon}_{m} + \left(1-\hat{\epsilon}_{m}\right)\bigg( \int_{\gamma_{0}}^{\infty}\left(f\left(\hat{\gamma}_{m},\hat{\epsilon}_{m}\right)\right)^{-\Theta_{m}}\cdot p_{\hat{\gamma}_{m}}(x)\mathrm{d}x + \int_{0}^{\gamma_{0}}p_{\hat{\gamma}_{m}}(x)\mathrm{d}x \bigg)\\
          & = \mathcal{H}_{m}\left(N_{LoP},\boldsymbol{\mathrm{L}}_{m},\hat{\epsilon}_{m}\right) + \mathcal{G}_{m}\left(N_{LoP},\boldsymbol{\mathrm{L}}_{m},\hat{\epsilon}_{m}\right).
       \end{aligned}
    \end{equation}
    \vspace{-0.2cm}
    where
    \begin{equation*}
       \begin{cases}
         \mathcal{H}_{m}\left(N_{LoP},\boldsymbol{\mathrm{L}}_{m},\hat{\epsilon}_{m}\right) \triangleq \hat{\epsilon}_{m} + \left(1-\hat{\epsilon}_{m}\right)\int_{0}^{\gamma_{0}}p_{\hat{\gamma}_{m}}(x)\mathrm{d}x, \\
         \mathcal{G}_{m}\left(N_{LoP},\boldsymbol{\mathrm{L}}_{m},\hat{\epsilon}_{m}\right) \triangleq \int_{\gamma_{0}}^{\infty}\left(f\left(\hat{\gamma}_{m},\hat{\epsilon}_{m}\right)\right)^{-\Theta_{m}}\cdot p_{\hat{\gamma}_{m}}(x)\mathrm{d}x.
       \end{cases}
    \end{equation*}
    \vspace{-0.3cm}
    \par When in the high-end SINR region, i.e., $\hat{\gamma}_{m} \gg 1$, we obtain that $\sqrt{\frac{\hat{\gamma}_{m}^{2}+\hat{\gamma}_{m}}{\left(1+\hat{\gamma}_{m}\right)^{2}}} \approx 1$ and $\sqrt{\frac{\hat{\gamma}_{m}^{2}+\hat{\gamma}_{m}}{\left(1+\hat{\gamma}_{m}\right)^{2}}}Q^{-1}(\hat{\epsilon}_{m}) \approx 0$. Correspondingly, the expression of $\mathcal{G}_{m}\left(N_{LoP},\boldsymbol{\mathrm{L}}_{m},\hat{\epsilon}_{m}\right)$ can be reformulated by (\ref{e35}),
    \begin{figure*}[b]
    \setstretch{0.88}
    \centering
    \hrulefill
    \begin{equation}\label{e35}
       \begin{aligned}
           \mathcal{G}_{m} & = \int_{\gamma_{0}}^{\infty}\biggl(\frac{1+\hat{\gamma}_{m}}{exp\left\{\sqrt{\frac{\hat{\gamma}_{m}^{2}+2\hat{\gamma}_{m}}{\left(1+\hat{\gamma}_{m}\right)^{2}}}   Q^{-1}\left(\hat{\epsilon}_{m}\right)\right\}}\biggl)^{-\Theta_{m}}\cdot p_{\hat{\gamma}_{m}}\left(x\right)\mathrm{d}x \overset{(a)}{\approx} \int_{\gamma_{0}}^{\infty} \left(1+\hat{\gamma}_{m}\right)^{-\Theta_{m}}\cdot p_{\hat{\gamma}_{m}}\left(x\right)\mathrm{d}x \\
           &\overset{(b)}{\approx} \int_{\gamma_{0}}^{\infty} e^{-\Theta_{m}\ln\left(\hat{\gamma}_{m}\right)}\cdot p_{\hat{\gamma}_{m}}\left(x\right)\mathrm{d}x \overset{(c)}{\leq} \int_{\gamma_{0}}^{\infty} e^{-\Theta_{m}\left(\hat{\gamma}_{m}-1\right)}\cdot p_{\hat{\gamma}_{m}}\left(x\right)\mathrm{d}x \overset{(d)}{\approx} \int_{0}^{\infty} e^{-\Theta_{m}\hat{\gamma}_{m}}\cdot p_{\hat{\gamma}_{m}}\left(x\right)\mathrm{d}x \\
           & = \mathbb{E}_{\hat{\gamma}_{m}}\left[e^{-\Theta_{m}\hat{\gamma}_{m}}\right]= \overline{\mathbb{M}}_{\hat{\gamma}_{m}}\left(\Theta_{m}\right)
       \end{aligned}
    \end{equation}
    \hrulefill
    \end{figure*}
    where $(a)$, $(b)$, and $(d)$ can be directly obtained from $\hat{\gamma}_{m} \gg 1$, $(c)$ is obtained by using the first-order Taylor expansion.
    \par By combining (\ref{e33})-(\ref{e35}), we can finally derive that
    \begin{equation}\label{e36}
       \begin{aligned}
          \overline{\mathbb{M}}_{s_{k}}\!\left(\theta_{k}\right) \cong \mathcal{H}_{k}\!\left(N_{LoP},\boldsymbol{\mathrm{L}}_{m},\hat{\epsilon}_{m}\right) + \left(1\!-\!\hat{\epsilon}_{m}\right)\overline{\mathbb{M}}_{\hat{\gamma}_{m}}\!\left(\Theta_{m}\right).
       \end{aligned}
    \end{equation}
    \par So the proof of \textbf{\emph{Theorem 3}} is concluded.
\end{proof}

\par By combining \textrm{Theorem 1, 2, and 3}, the closed-from expression of UB-SDVP can be approximated by the following theorem.
\begin{theorem}
   Given QoS exponent $\theta_{m}$ and decoding error probability $\hat{\epsilon}_{m}$, the UB-SDVP of $\mathbb{P}\left(W_{m}\left(t\right) \geq d_{th} \right)$ for each mobile UE $m$ can be approximated as follows:
   \vspace{-0.2cm}
   \begin{equation}\label{e37}
      \begin{aligned}
        \!\!\!\! & \mathbb{P}\left(W_{m}\left(t\right) > d_{th} \right) \cong \frac{\big(\mathcal{H}_{m}\!\left(N_{LoP},\boldsymbol{\mathrm{L}}_{m},\!\hat{\epsilon}_{m}\right) \!+\! \left(1\!-\!\hat{\epsilon}_{m}\right)\!\overline{\mathbb{M}}_{\hat{\gamma}_{m}}\!\left(\!\Theta_{m}\!\right)\!\!\big)^{d_{th}}}{1\!-\!e^{\lambda_{m}^{\dag}\left(e^{\theta_{m}}-1\right)}\left(\!\mathcal{H}_{k}\!\left(N_{LoP},\boldsymbol{\mathrm{L}}_{m},\!\hat{\epsilon}_{m}\right) \!+\! \left(1\!-\!\hat{\epsilon}_{m}\!\right)\!\overline{\mathbb{M}}_{\hat{\gamma}_{m}}\!\!\left(\!\Theta_{m}\!\right)\!\right)}.
      \end{aligned}
   \end{equation}
\end{theorem}
\begin{proof}
   According to Theorem 1, 2, 3 and (\ref{e32}), Theorem 4 can be easily proved.
\end{proof}
\vspace{-0.3cm}
\section{Minimizing Upper-Bounded Statistical Delay Violation Probability}
\par In Sec. III, the promoted MGF-SNC is employed to derive the closed-form expression of UB-SDVP. However, the decoding error probability is tightly correlated with the channel conditions, namely $\hat{\gamma}_{m}\left(N_{LoP},\boldsymbol{\mathrm{L}}_{m}\right)$ \cite{zeng2019enabling,ren2020joint,zhang2020statistical0}. In this section, we formulate and thoroughly investigate the UB-SDVP minimization problem, in which the decoding error probability is considered as a bivariate function with respect to $\left\{N_{LoP},\boldsymbol{\mathrm{L}}_{m}\right\}$.
\vspace{-0.5cm}
\subsection{Minimizing the Decoding Error Probability}
\par According to (\ref{e16}), given the target maximum achievable data rate $r_{m}$, the decoding error probability for each mobile UE $m$ can be closely approximated as follows:
\vspace{-0.2cm}
\begin{equation}\label{e38}
   \!\!\hat{\epsilon}_{m}\big(N_{LoP},\!\boldsymbol{\mathrm{L}}_{m}\!\big) \!=\! Q\!\!\left(\!\!\!\sqrt{\!\!\frac{R_{CU}\!\!\left(C\big(\hat{\gamma}_{m}(N_{LoP},\!\boldsymbol{\mathrm{L}}_{m})\big)\!-\!r_{m}\!\right)^{2}}{\mathcal{V}\big(\hat{\gamma}_{m}(N_{LoP},\!\boldsymbol{\mathrm{L}}_{m})\big)}}\!\right)\!\!.
\end{equation}
\vspace{-1.5em}
\par The expectation of the decoding error probability $\hat{\epsilon}_{m}\big(N_{LoP},\boldsymbol{\mathrm{L}}_{m}\big)$ with respect to SINR $\hat{\gamma}_{m}$ can be expressed as follows:
\vspace{-0.2cm}
\begin{equation}\label{e39}
   \begin{aligned} 
      & \boldsymbol{\mathcal{E}}_{m}\left(N_{LoP},\boldsymbol{L}_{m}\right) = \mathbb{E}_{\hat{\gamma}_{m}}\left[\hat{\epsilon}_{m}\big(N_{LoP},\boldsymbol{\mathrm{L}}_{m}\big)\right] = \int_{0}^{\infty} \hat{\epsilon}_{m}\big(N_{LoP},\boldsymbol{\mathrm{L}}_{m}\big)\cdot p_{\hat{\gamma}_{m}\left(N_{LoP},\boldsymbol{\mathrm{L}}_{m}\right)} \left(x\right)\mathrm{d} x,
   \end{aligned}
\end{equation}
where $\mathbb{E}_{\hat{\gamma}_{m}}\!\left[\ \cdot\ \right]$ represents the expectation operation over SINR $\hat{\gamma}(N_{LoP},\boldsymbol{\mathrm{L}}_{m})$.
\par According to Theorem 2, the minimization problem for the upper-bounded statistical delay violation probability can be formulated as follows:
\vspace{-0.3cm}
\begin{subequations}\label{e40}
   \begin{align}
     \mathcal{P}1:\ \  \hat{\boldsymbol{\epsilon}}^{\star}\!\!\left(\boldsymbol{\mathrm{L}}^{\star},N_{LoP}^{\star}\right) & = \!\arg\!\!\min_{\{\boldsymbol{\mathrm{L}}, N_{LoP}\!\}}\!\!\bigg\{\!\sum\limits_{m=1}^{M} \frac{\left(\overline{\mathbb{M}}_{s_{m}}\!\!\left(\theta_{m}\right)\right)^{d_{th}}}{1-\mathbb{M}_{a_{m}}\!\!\left(\theta_{m}\right)\overline{\mathbb{M}}_{s_{m}}\!\!\left(\theta_{m}\right)} \bigg\} \\
     & = \arg\min_{\{\boldsymbol{\mathrm{L}}, N_{LoP}\}}\bigg\{\sum\limits_{m=1}^{M} \overline{\mathbb{M}}_{s_{m}}\left(\theta_{m}\right) \bigg\},
   \end{align}
\end{subequations}
where $\boldsymbol{\mathrm{L}} \triangleq \left\{\boldsymbol{\mathrm{L}}_{1},\cdots,\boldsymbol{\mathrm{L}}_{M}\right\}$ denotes the linear detectors matrix vector of mobile UEs. $\hat{\boldsymbol{\epsilon}}^{\star}\left(\boldsymbol{\mathrm{L}}^{\star},N_{LoP}^{\star}\right) \triangleq \left\{\hat{\epsilon}_{1}^{\star}\left(\boldsymbol{\mathrm{L}}_{1}^{\star},N_{LoP}^{\star}\right),\cdots,\hat{\epsilon}_{M}^{\star}\left(\boldsymbol{\mathrm{L}}_{M}^{\star},N_{LoP}^{\star}\right)\right\}$ represents the minimum decoding error probabilities of the $M$ mobile UEs corresponding to the optimal linear detector matrix vector $\boldsymbol{\mathrm{L}}^{\star}$ and the optimal pilot length $N_{LoP}^{\star}$.
\vspace{-0.1cm}
\par According to Theorem 3, the inverse-MGF of the service process for each mobile UE $m$, denoted by $\overline{\mathbb{M}}_{s_{m}}\left(\theta_{m}\right)$, can be denoted as follows:
\vspace{-0.1cm}
\begin{equation}\label{e41}
    \begin{aligned}
      \overline{\mathbb{M}}_{s_{m}}\!\!\left(\theta_{m}\!\right) \triangleq  \left(1-e^{-\theta_{m}R_{CU}r_{m}}\right)\boldsymbol{\mathcal{E}}_{m}\left(N_{LoP},\boldsymbol{L}_{m}\right) + e^{-\theta_{m}R_{CU}r_{m}}.
    \end{aligned}
\end{equation}
\vspace{-0.3cm}
\par Substituting (\ref{e41}) back into (\ref{e40}b), problem $\mathcal{P}1$ can be converted into the following equivalent minimization problem, denoted by $\mathcal{P}2$, as follows:
\vspace{-1.0em}
\begin{equation}\label{e42}
  \begin{aligned}
    & \mathcal{P}2:\ \hat{\boldsymbol{\epsilon}}^{\star}\left(\boldsymbol{\mathrm{L}}^{\star},N_{LoP}^{\star}\right) = \arg\min_{\{\boldsymbol{\mathrm{L}}, N_{LoP}\}}\!\!\!\bigg\{\sum\limits_{m=1}^{M} \!e^{-\theta_{m}R_{CU}r_{m}} \!+\! \left(1\!-\!e^{-\theta_{m}R_{CU}r_{m}}\right)\!\!\boldsymbol{\mathcal{E}}_{m}\left(N_{LoP},\boldsymbol{\mathrm{L}}_{m}\right) \!\!\bigg\}.
  \end{aligned}
\end{equation}
\vspace{-1.0em}
\par The structure of problem $\mathcal{P}2$ reveals that the monotonicity of decoding error probability $\hat{\epsilon}_{m}\left(\boldsymbol{N_{LoP},\mathrm{L}}_{m}\right)$ of great importance in analyzing the convexity of the minimization problem $\mathcal{P}2$; this motivates Theorem 5, as follows:
\begin{theorem}
  The decoding error probability $\hat{\epsilon}_{m}\big(N_{LoP},\boldsymbol{\mathrm{L}}_{m}\big)$ of each mobile UE $m$ is a strictly monotonically decreasing function with respect to SINR $\hat{\gamma}_{m}(N_{LoP},\boldsymbol{\mathrm{L}}_{m})$ for a given maximum achievable data rate $r_{m}$.
\end{theorem}

\begin{proof}
  We define the function of $\hat{\gamma}_{m}(N_{LoP},\boldsymbol{\mathrm{L}}_{m})$ as follows:
  \begin{equation*}
    \begin{aligned}
      & g\big(\hat{\gamma}_{m}\big) = \sqrt{\frac{R_{CU}\left(C\big(\hat{\gamma}_{m}(N_{LoP},\boldsymbol{\mathrm{L}}_{m})\big)-r_{m}\right)^{2}}{\mathcal{V}\big(\hat{\gamma}_{m}(N_{LoP},\boldsymbol{\mathrm{L}}_{m})\big)}}.
    \end{aligned}
  \end{equation*}
  \vspace{-0.2cm}
  \par The first-order partial derivative of $\hat{\epsilon}_{m}\big(N_{LoP},\boldsymbol{\mathrm{L}}_{m}\big)$ with respect to $\hat{\gamma}_{m}(N_{LoP},\boldsymbol{\mathrm{L}}_{m})$ is expressed as
  \begin{equation*}
    \frac{\partial \hat{\epsilon}_{m}}{\partial \hat{\gamma}_{m}} = -\frac{1}{\sqrt{2\pi}}e^{-g^{2}\big(\hat{\gamma}_{m}\left(N_{LoP},\boldsymbol{\mathrm{L}}_{m}\right)\big)}\frac{\partial g\big(\hat{\gamma}_{m}(N_{LoP},\boldsymbol{\mathrm{L}}_{m})\big)}{\partial \hat{\gamma}_{m}(N_{LoP},\boldsymbol{\mathrm{L}}_{m})}
  \end{equation*}
  \par In addition, the first-order partial derivative of function $g\big(\hat{\gamma}_{m}\left(N_{LoP},\boldsymbol{\mathrm{L}}_{m}\right)\big)$ with respect to $\hat{\gamma}_{m}\left(N_{LoP},\boldsymbol{\mathrm{L}}_{m}\right)$ is given as the following equation:
  \vspace{-0.2cm}
  \begin{equation*}
     \begin{aligned}
         & \frac{\partial g\big(\hat{\gamma}_{m}\left(N_{LoP},\boldsymbol{\mathrm{L}}_{m}\right)\big)}{\partial \hat{\gamma}_{m}(N_{LoP},\boldsymbol{\mathrm{L}}_{m})}  = R_{CU}\frac{1-\frac{\ln\left(1+\hat{\gamma}_{m}\right)-r_{m}\ln2}{\left(1+\hat{\gamma}_{m}\right)^{2}-1}}{\sqrt{\left(1+\hat{\gamma}_{m}\right)^{2}-1}}\\
         & \geq R_{CU}\frac{1-\frac{\ln\left(1+\hat{\gamma}_{m}\right)}{\left(1+\hat{\gamma}_{m}\right)^{2}-1}}{\sqrt{\left(1+\hat{\gamma}_{m}\right)^{2}-1}} = R_{CU}\frac{1-\frac{\ln\left(1+\hat{\gamma}_{m}\right)}{\left(\hat{\gamma}_{m}+2\right)\hat{\gamma}_{m}}}{\sqrt{\left(1+\hat{\gamma}_{m}\right)^{2}-1}}.
     \end{aligned}
  \end{equation*}
  \vspace{-0.2cm}
  \par Define another function, namely $\tilde{g}\left(\hat{\gamma}_{m}\right) = \frac{\ln\left(1+\hat{\gamma}_{m}\right)}{\hat{\gamma}_{m}\left(\hat{\gamma}_{m}+2\right)}$, it is easy to prove that $\tilde{g}\left(\hat{\gamma}_{m}\right)$ is a decreasing function with respect to $\hat{\gamma}_{m}$. By exploiting the Lhobita's Law, we can obtain that $\max\{\tilde{g}\left(\hat{\gamma}_{m}\right)\} = \tilde{g}\left(\hat{\gamma}_{m}\right)\big|_{\hat{\gamma}_{m} \rightarrow 0} = \frac{1}{2}$, and $\min\{\tilde{g}\left(\hat{\gamma}_{m}\right)\} = \tilde{g}\left(\hat{\gamma}_{m}\right)\big|_{\hat{\gamma}_{m} \rightarrow \infty} = 0$. Therefore, we can conclude that $\frac{\partial g\big(\hat{\gamma}_{m}\left(N_{LoP},\boldsymbol{\mathrm{L}}_{m}\right)\big)}{\partial \hat{\gamma}_{m}\left(N_{LoP},\boldsymbol{\mathrm{L}}_{m}\right)} > 0$. So the decoding error probability $\hat{\epsilon}_{m}\big(\hat{\gamma}_{m}\left(N_{LoP},\boldsymbol{\mathrm{L}}_{m}\right)\big)$ is a monotonically decreasing function with respect to $\hat{\gamma}_{m}\left(N_{LoP},\boldsymbol{\mathrm{L}}_{m}\right)$.
  \par So the proof of \textbf{\emph{Theorem 5}} is concluded.
\end{proof}

\par Theorem 5 refers that $\mathcal{P}2$ can be solved from two aspects. On the one hand, we regard the linear detector matrix as a function with respect to pilot length and derive the expression of the optimal linear detector matrix, i.e., $\boldsymbol{\mathrm{L}}^{\star}(N_{LoP})$. On the other hand, since the pilot length $N_{LoP}$ is an integer in the interval $\left[M,N_{CU}-1\right]$, $\mathcal{P}2$ is degenerated into a one-dimensional integer-search problem once the expression of $\boldsymbol{\mathrm{L}}^{\star}(N_{LoP})$ is determined.

\subsubsection{Subproblem 1}
\par Given a pilot length $N_{LoP}$, the minimization problem $\mathcal{P}2$ can be reformulated as follows:
\vspace{-0.2cm}
\begin{equation}\label{e43}
  \mathcal{P}3:\ \boldsymbol{\mathrm{L}}^{\star}(N_{LoP}) = \arg\min_{\{\boldsymbol{\mathrm{L}}\}} \bigg\{\sum\limits_{m=1}^{M}\boldsymbol{\mathcal{E}}_{m}\left(N_{LoP},\boldsymbol{\mathrm{L}}_{m}\right)\bigg\}.
\end{equation}
\par In problem $\mathcal{P}3$, the optimal linear detector matrix $\boldsymbol{\mathrm{L}}^{\star} \triangleq \left\{\boldsymbol{\mathrm{L}}_{1}^{\star},\cdots,\boldsymbol{\mathrm{L}}_{M}^{\star}\right\}$ can be obtained by the theorem as follows:
\begin{theorem}
  Given a pilot length $N_{LoP}$, the optimal linear detector matrix for the $M$ mobile UEs can be denoted as $\boldsymbol{L}^{\star} \triangleq \left\{\boldsymbol{\mathrm{L}}_{1}^{\star},\cdots,\boldsymbol{\mathrm{L}}_{M}^{\star}\right\}$, where $\boldsymbol{\mathrm{L}}_{m}^{\star} = \boldsymbol{\mathrm{d}}_{m}^{H}\boldsymbol{\mathrm{G}}^{-1}$, $m \in \mathcal{M}$.
\end{theorem}

\begin{proof}
  \textbf{\emph{Theorem 5}} shows that $\hat{\epsilon}_{m}\left(N_{LoP},\boldsymbol{\mathrm{L}}_{m}\right)$ is a strictly monotonically decreasing function with respect to $\hat{\gamma}_{m}\left(N_{LoP},\boldsymbol{\mathrm{L}}_{m}\right)$. Moreover, the function $p_{\hat{\gamma}_{m}(N_{LoP},\boldsymbol{\mathrm{L}}_{m})}(x)$ represents the probability when the value of SINR $\hat{\gamma}_{m}(N_{LoP},\boldsymbol{\mathrm{L}}_{m})$ is $x$. Hence, for a given pilot length $N_{LoP}$, minimizing $\boldsymbol{\mathcal{E}}_{m}\left(N_{LoP},\boldsymbol{\mathrm{L}}_{m}\right)$ is equivalent to maximizing $\hat{\gamma}_{m}(N_{LoP},\boldsymbol{\mathrm{L}}_{m})$. According to (13), we have
   \begin{equation}\label{e44}
     \begin{aligned}
      & \hat{\gamma}_{m}\left(N_{LoP},\boldsymbol{\mathrm{L}}_{m}\right) = \frac{\boldsymbol{\mathrm{L}}_{m}\boldsymbol{\mathrm{d}}_{m}\boldsymbol{\mathrm{d}}_{m}^{H}\boldsymbol{\mathrm{L}}_{m}^{H}}{\boldsymbol{\mathrm{L}}_{m} \big(\boldsymbol{\mathrm{G}} - \boldsymbol{\mathrm{d}}_{m}\boldsymbol{\mathrm{d}}_{m}^{H}\big)\boldsymbol{\mathrm{L}}_{m}^{H}}\\
      & = \frac{1}{\frac{\boldsymbol{\mathrm{L}}_{m}\boldsymbol{\mathrm{G}}\boldsymbol{\mathrm{L}}_{m}^{H}}{\boldsymbol{\mathrm{L}}_{m}\boldsymbol{\mathrm{d}}_{m}\boldsymbol{\mathrm{d}}_{m}^{H}\boldsymbol{\mathrm{L}}_{m}^{H}}-1} = \frac{1}{\frac{\|\boldsymbol{\mathrm{L}}_{m}\boldsymbol{\mathrm{G}}^{\frac{1}{2}}\|^{2}}{\|\boldsymbol{\mathrm{L}}_{m}\boldsymbol{\mathrm{G}}^{\frac{1}{2}}\boldsymbol{\mathrm{G}}^{-\frac{1}{2}}\boldsymbol{\mathrm{d}}_{m}\|^{2}}-1}\\
      & \overset{(a)}{\leq} \frac{\|\boldsymbol{\mathrm{G}}^{-\frac{1}{2}}\boldsymbol{\mathrm{d}}_{m}\|^{2}}{1-\|\boldsymbol{\mathrm{G}}^{-\frac{1}{2}}\boldsymbol{\mathrm{d}}_{m}\|^{2}} = \frac{\boldsymbol{\mathrm{d}}_{m}^{H}\boldsymbol{\mathrm{G}}^{-1}\boldsymbol{\mathrm{d}}_{m}}{1-\boldsymbol{\mathrm{d}}_{m}^{H}\boldsymbol{\mathrm{G}}^{-1}\boldsymbol{\mathrm{d}}_{m}},
     \end{aligned}
  \end{equation}
  where the inequality $(a)$ is obtained by using Cauchy-Schwarz inequality. If and only if $\boldsymbol{\mathrm{d}}_{m}^{H}\boldsymbol{\mathrm{G}}^{-1/2} = \boldsymbol{\mathrm{L}}_{m}\boldsymbol{\mathrm{G}}^{1/2}$ holds, the inequality $(a)$ takes the equal sign. Therefore, the optimal solution of $\mathcal{P}3$ can be denoted as $\boldsymbol{\mathrm{L}}^{\star} \triangleq \left\{\boldsymbol{\mathrm{L}}_{1}^{\star},\cdots,\boldsymbol{\mathrm{L}}_{M}^{\star}\right\}$, where $\boldsymbol{L}_{m}^{\star} = \boldsymbol{\mathrm{d}}_{m}^{H}\boldsymbol{\mathrm{G}}^{-1}$, $m \in \mathcal{M}$.
  \par So the proof of \textbf{\emph{Theorem 6}} is concluded.
\end{proof}

\subsubsection{Subproblem 2}
  According to Theorem 6, we consider $\boldsymbol{\mathrm{L}}^{\star}(N_{LoP})$ to be a function of the $N_{LoP}$. As a result, the problem $\mathcal{P}2$ degenerates into a one-dimensional integer search problem of the form:
  \vspace{-0.2cm}
  \begin{equation}\label{e45}
   \begin{aligned}
     \mathcal{P}4:\ \ N_{LoP}^{\star}\! = \! \arg\!\min_{M \leq N\!_{LoP} \leq N\!_{CU}-1}\! \bigg\{\sum\limits_{m=1}^{M}\!\!\boldsymbol{\mathcal{E}}\!_{m}\!\left(N\!_{LoP},\!\boldsymbol{\mathrm{L}}^{\star}_{m}\!(N\!_{LoP}\!)\!\right) \bigg\}.
   \end{aligned}
  \end{equation}
\par To efficiently address $\mathcal{P}4$, we propose a variant of the Golden-Section search method (GSS) named integer-form Golden-Section search algorithm (IFGSS). Unlike the traditional GSS, the lower-search and upper-search bounds of the IFGSS are determined by the nearest integers to the respective Golden-Section points in each iteration. We provide a detailed description of the IFGSS in Algorithm 1.

\begin{algorithm}[t]
\setstretch{0.85}
\caption{IFGSS Algorithm}
\begin{algorithmic}[1]
\REQUIRE 
Lower search bound $M$; Upper search bound $N_{CU}-1$; Golden-section ratio $\tau = 0.618$; Maximum achievable rate $\boldsymbol{r}=\left\{r_{1},\cdots,r_{M}\right\}$ ;\\
\ENSURE
Optimal length of pilot $N_{LoP}^{\star}$;\\ 
\STATE Set $lower_{1} = K$, $upper_{1} = M$;\\
\STATE Set initial search interval $\left[lower_{1},upper_{1}\right]$;\\
\STATE Compute $N_{1}^{1} = upper_{1} - \lfloor \tau\left(upper_{1}-lower_{1}\right) \rfloor$;\\
\STATE Compute $N_{2}^{1} = lower_{1} + \lceil \tau\left(upper_{1}-lower_{1}\right) \rceil$;\\
\STATE Set iteration index $i = 1$; \\
\IF{$\sum\limits_{m=1}^{M} \boldsymbol{\mathcal{E}}_{m}\left(N_{2}^{i},\boldsymbol{L}_{m}^{\star}(N_{2}^{i})\right) >\sum\limits_{m=1}^{M} \boldsymbol{\mathcal{E}}_{m}\left(N_{1}^{i},\boldsymbol{L}_{m}^{\star}(N_{1}^{i})\right)
 $}
\STATE $lower_{i+1} \leftarrow lower_{i}$, $upper_{i+1} \leftarrow N_{2}^{i}$;\\
\STATE $N_{2}^{i+1} \leftarrow N_{1}^{i}$;\\
\STATE $N_{1}^{i+1} \leftarrow upper_{i+1} - \lfloor \tau \left(upper_{i+1} - lower_{i+1}\right)\rfloor$;\\
\ELSE
\STATE $lower_{i+1} \leftarrow N_{1}^{i}$, $upper_{i+1} \leftarrow upper_{i}$;\\
\STATE $N_{1}^{i+1} \leftarrow N_{2}^{i}$;\\
\STATE $N_{2}^{i+1} \leftarrow lower_{i+1} + \lceil \tau\left(upper_{i+1} - lower_{i+1}\right) \rceil$;\\
\ENDIF
\IF{$upper_{i+1} == lower_{i+1}$}
\STATE \textbf{return} $N_{LoP}^{\star} = N_{1}^{i+1}$;\\ 
\ELSE
\STATE Update $i:=i+1$;\\
\STATE \textbf{go to} Step \textbf{6};\\
\ENDIF
\label{code:recentEnd}
\end{algorithmic}
\end{algorithm}

\par According to Theorem 6 and Algorithm 1, the minimum precoding error probability $\boldsymbol{\mathcal{E}}_{m}\big(N_{LoP}^{\star},$ $\boldsymbol{\mathrm{L}}_{m}^{\star}\left(N_{LoP}^{\star}\right)\big)$, $m\in\mathcal{M}$ can be determined, and the optimal linear detector matrix $\boldsymbol{\mathrm{L}}^{\star}$ is given as follows:
\vspace{-0.1cm}
\begin{equation}\label{e46}
  \begin{aligned}
    \boldsymbol{\mathrm{\mathrm{L}}}^{\star}\!(N_{LoP}^{\star}) \!=\! \left\{\boldsymbol{\mathrm{L}}_{1}^{\star}(N_{LoP}^{\star}),\!\cdots\!,\boldsymbol{\mathrm{L}}_{M}^{\star}(N_{LoP}^{\star})\right\} \!=\! \left(\boldsymbol{\mathrm{D}}^{\star}\right)^{H}\!\!\left(\boldsymbol{\mathrm{G}}^{\star}\right)^{-1}\!\!,
  \end{aligned}
\end{equation}
where $\boldsymbol{\mathrm{D}}^{\star} = \boldsymbol{\mathrm{\delta}}^{\star}\hat{\boldsymbol{\mathrm{H}}}$, $\boldsymbol{\mathrm{G}}^{\star} = \boldsymbol{\mathrm{D}}^{\star}\left(\boldsymbol{\mathrm{D}}^{\star} \right)^{H} + \frac{1}{\omega^{\star}}\boldsymbol{\mathrm{I}}_{L}$. $\boldsymbol{\mathrm{\delta}}^{\star} = \mathrm{diag}\left(\delta_{1}^{\star},\cdots,\delta_{M}^{\star}\right)$, $\delta_{m}^{\star} = \frac{\rho N_{LoP}^{\star} \lambda_{m} \beta_{m}}{\rho N_{LoP}^{\star} \lambda_{m} \beta_{m} + 1}$, and $\omega^{\star} = \big(\sum\limits_{i=1}^{M}\frac{\lambda_{i}\beta_{i}}{\rho N_{LoP}^{\star} \lambda_{i}\beta_{i} + 1} + \frac{1}{\rho} \big)^{-1}$.

\subsubsection{Computational Complexity Analysis}
\par Note that subproblem $\mathcal{P}4$ can also be solved using the exhaustive method (EM), whose computational complexity is $\mathcal{O}\left(N_{CU}-M\right)$. The proposed IFGSS algorithm, however, has a computational complexity of $\mathcal{O}\left(\log_{2}\left(N_{CU} - M\right)\right)$ \cite{tsai2010golden}, and thus can achieve superior convergence performance compared to EM, which leads to inefficient search results when the blocklength $N_{CU}$ is large.
\vspace{-0.1cm}
\section{Performance Optimizations of EP-EC and EP-EE}
\vspace{-0.1cm}
\par In Sec. IV, based on the theoretical framework, we have investigated and solved the UB-SDVP and decoding error probability minimization problems with an extremely small delay. However, the performance tradeoffs of xURLLC are equally essential \cite{she2021tutorial,park2022extreme}. In this section, to characterize the tail distribution (i.e., effective capacity (EC)\cite{guo2019resource, amjad2019effective, hou2018burstiness}) and energy efficiency \cite{she2021tutorial,bennis2018ultrareliable,park2022extreme} in the short-packet data communications of xURLLC, we propose two novel concepts, known as EP-EC and EP-EE. Then, we consider the optimal solutions $\left\{N_{LoP}^{\star}(\rho),\boldsymbol{\mathrm{L}^{\star}}(\rho)\right\}$ obtained in Sec. IV as functions of transmit power $\rho$ and investigate the performance optimizations for EP-EC and EP-EE.

\vspace{-0.3cm}
\subsection{The Concepts of EP-EC and EP-EE}
\par Effective capacity refers to the maximum arrival rate that guarantees the corresponding QoS requirements for a given service rate\cite{guo2019resource, amjad2019effective, hou2018burstiness}.
Given the cumulative service process $S_{m}\left(t\right)=\sum\limits_{i=1}^{t} s_{m}(i)$ and QoS exponent $\theta_{m}$, the effective capacity for each mobile UE $m$, denoted by $EC_{m}(\theta_{m})$, can be expressed as follows:
\vspace{-0.1cm}
\begin{equation}\label{e47}
  EC_{m}\left(\theta_{m}\right) \triangleq -\lim_{t \rightarrow \infty}\frac{1}{t\theta_{m}}\log\left(\mathbb{E}_{\hat{\gamma}_{m}}\left[e^{-\theta_{m}S_{m}\left(t\right)}\right]\right).
\end{equation}
\par By combining (\ref{e41}) and (\ref{e47}), since the $S_{m}(t)$ is uncorrelated across different time slots\cite{guo2019resource}, we can extend EC to the finite blocklength regime by the proposed theoretical framework. Specifically, we can define EP-EC $\mathcal{EC}_{m}\left(\theta_{m}\right)$ using the inverse-MGF over the service rate $s_{m}$ as follows:
\vspace{-0.2cm}
\begin{equation}\label{e48}
   \begin{aligned}
     & \mathcal{EC}_{m}\!\left(\theta_{m}\right) \!\triangleq\!  -\frac{1}{\theta_{m}}\!\log\!\big[\overline{\mathbb{M}}_{s_{m}}\left(\theta_{m}\right)\big] \!=\! -\frac{1}{\theta_{m}}\!\log\!\!\bigg[e^{-\theta_{m}R_{CU}r_{m}} \!+\! \left(1\!-\!e^{-\theta_{m}R_{CU}r_{m}}\!\right)\!\boldsymbol{\mathcal{E}}_{m}\!\left(N_{LoP},\boldsymbol{\mathrm{L}}_{m}\!\right)\!\bigg].
   \end{aligned}
\end{equation}
\par From (\ref{e48}), EP-EC can intuitively reveals the key factors affecting the performance of short-packet communications for xURLLC, including QoS exponent $\theta_{m}$, maximum achievable data rate $r_{m}$, pilot length $N_{LoP}$, blocklength $N_{CU}$, linear detector matrix $\boldsymbol{\mathrm{L}}_{m}$.

\par Next, we consider the most common power consumption model, which is defined as
\vspace{-0.2cm}
\begin{equation}\label{e49}
   P_{tot} = P_{c}+ \frac{1}{\varphi} \cdot \rho,
\end{equation}
\vspace{-0.1cm}
where $P_{c}$ denotes the constant circuit power, which corresponds to the power consumption of the transmitter circuitry, $0 \leq \varphi \leq 1$ denotes the power amplifier efficiency.

\par According to (\ref{e48}) and (\ref{e49}), we propose the concept of EP-EE to characterize the energy efficiency of the developed xURLLC-enabled massive MU-MIMO wireless networks in the finite blocklength regime as the ratio of the EP-EC to the total power consumption as follows:
\vspace{-0.2cm}
\begin{equation}\label{e50}
   \begin{aligned}
     \vartheta_{m}\left(\theta_{m}\right) = \mathcal{EC}_{m}\left(\theta_{m}\right)/P_{tot}.
   \end{aligned}
\end{equation}
\vspace{-1.0cm}
\subsection{The EP-EC Maximization Problem} 
\vspace{-0.1cm}
\par According to (\ref{e48}), a larger EP-EC implies that the developed xURLLC-enabled massive MU-MIMO wireless networks can support a higher total arrival rate while guaranteeing the expected statistical QoS provisioning. To maximize the EP-EC, we formulate the following optimization problem:
\vspace{-1em}
\begin{equation}\label{e51}
  \begin{aligned}
    \mathcal{P}5:\ \max_{\left\{N_{LoP}(\rho),\boldsymbol{\mathrm{L}}(\rho)\right\}} \biggl\{\sum\limits_{m=1}^{M}-\frac{1}{\theta_{m}}\log\!\!\bigg[e^{-\theta_{m}R_{CU}r_{m}} \!+\! \left(1\!-\!e^{-\theta_{m}R_{CU}r_{m}}\right)\!\boldsymbol{\mathcal{E}}_{m}\big(N_{LoP}(\rho),\boldsymbol{\mathrm{L}}_{m}(\rho)\big)\bigg]\!\biggl\}.
  \end{aligned}
\end{equation}
\vspace{-0.3cm}
\par Since the function $\log[\boldsymbol{\cdot}]$ is a monotonically increasing function, it can be shown that problem $\mathcal{P}5$ is equivalent to problem $\mathcal{P}1$, and shares the identical optimal solutions with $\mathcal{P}1$, i.e., $\left\{N_{LoP}^{\star}(\rho),\boldsymbol{\mathrm{L}}^{\star}(\rho)\right\}$.
\vspace{-0.3cm}
\subsection{The EP-EE Maximization Problem}
\vspace{-0.1cm}
\par Well known in communications, mindlessly increasing the transmit power for communication systems will typically result in excessive energy consumption rather than maximum energy efficiency \cite{rezvani2022optimal}. Consequently, from the perspective of green communication, maximizing EP-EE is essentially prominent for xURLLC\cite{park2022extreme,bennis2018ultrareliable,she2021tutorial}. Based on $\mathcal{P}5$, the EP-EE maximization problem for the developed xURLLC-enabled massive MU-MIMO networks in the finite blocklength regime can be formulated as follows:

\vspace{-0.4cm}
\begin{equation}\label{e52}
   \begin{aligned}
      \mathcal{P}6 :\quad & \max_{\rho > 0} \bigg\{\sum\limits_{m=1}^{M} \frac{\mathcal{EC}_{m}\left(N_{LoP}^{\star}(\rho),\boldsymbol{\mathrm{L}}_{m}^{\star}(\rho)\right)}{P_{c}+\frac{1}{\varphi}\cdot \rho}\bigg\},\\ 
      & s.t.\quad 0 \leq \rho \leq P_{max},
   \end{aligned}
\end{equation}
where $P_{max}$ denotes upper bound of the transmit power $\rho$.
\par For illustration purpose, we define two functions as follows:
\vspace{-0.4cm}
\begin{subequations}\label{e53}
  \begin{align}
    & \mathrm{EP}\text{-}\mathrm{EC}\left(\rho\right) \triangleq \sum\limits_{m=1}^{M} \mathcal{EC}_{m}\left(N_{LoP}^{\star}(\rho),\boldsymbol{L}_{m}^{\star}(\rho)\right),\\
    & P_{tot}\left(\rho\right) \triangleq P_{c} + \frac{1}{\varphi} \cdot \rho.
  \end{align}
\end{subequations}
\vspace{-0.4cm}
\par Then, the optimization problem $\mathcal{P}6$ can be reformulated as follows:
\vspace{-0.2cm}
\begin{equation}\label{e54}
   \mathcal{P}7:\quad \max_{0< \rho < P_{max}}\quad \vartheta(\rho) = \frac{\mathrm{EP}\text{-}\mathrm{EC}\left(\rho\right) }{P_{tot}\left(\rho\right)},
\end{equation}
where $\mathrm{EP}\text{-}\mathrm{EC}\left(\rho\right) \geq 0$ and $P_{tot}\left(\rho\right) > 0$.
\par By leveraging the fractional programming theory \cite{stancu2012fractional}, problem $\mathcal{P}7$ can be equivalently converted into a standard fractional programming problem as follows:
\begin{equation}\label{e55}
  \begin{aligned}
    \mathcal{P}8:\quad & \max_{0< \rho < P_{max}}\quad \vartheta(\rho),\\
     & s.t. \quad \mathrm{EP}\text{-}\mathrm{EC}\left(\rho\right)-\vartheta P_{tot}\left(\rho\right) \geq 0.
  \end{aligned}
\end{equation}
\par According to \cite{stancu2012fractional}, problem $\mathcal{P}8$ can be further reformulated as follows:
\begin{equation}\label{e56}
  \mathcal{P}9:\ F(\vartheta) = \max_{0< \rho < P_{max}} \left\{  \mathrm{EP}\text{-}\mathrm{EC}\left(\rho\right)-\vartheta P_{tot}\left(\rho\right) \right\}.
\end{equation}
\vspace{-2em}
\par $\mathcal{P}9$ is a bilateral optimization problem in which the parameter $\vartheta$ determines the relative weight between $\mathrm{EP}\text{-}\mathrm{EC}\left(\rho\right)$ and $P_{tot}\left(\rho\right)$. Since the objective of $\mathcal{P}7$ is to maximize $\mathrm{EP}\text{-}\mathrm{EC}\left(\rho\right)$ while minimizing $P_{tot}\left(\rho\right)$, $\mathcal{P}9$ and $\mathcal{P}7$ are two mutually equivalent optimization problems. To solve $\mathcal{P}9$ effectively, we propose two lemmas and a theorem that can be used to solve $\mathcal{P}9$ in advance.

\begin{lemma}
   $F(\vartheta) = \max\limits_{0< \rho < P_{max}} \left\{  \mathrm{EP}\text{-}\mathrm{EC}\left(\rho\right)-\vartheta P_{tot}\left(\rho\right) \right\}$ is a strictly decreasing and convex function with respect to the overall EP-EE $\vartheta$.
\end{lemma}

\begin{proof}
  Assume that $\rho_{2}^{\star}$ is the optimal transmit power that maximizes the function $F(\vartheta_{2})$. Then, for any $0 \leq \vartheta_{1}<\vartheta_{2}$, we can obtain that
  \vspace{-0.1cm}
  \begin{equation*}
    \begin{aligned}
      F(\vartheta_{2}) & = \!\! \max\limits_{0< \rho < P_{max}} \!\!\left\{\mathrm{EP}\text{-}\mathrm{EC}\!\left(\rho\right)\!-\!\vartheta_{2} P_{tot}\left(\rho\right)\right\}\\
     & = \mathrm{EP}\text{-}\mathrm{EC}\!\left(\rho_{2}^{\star}\right)\!-\!\vartheta_{2} P_{tot}\!\left(\rho_{2}^{\star}\right) \!<\! \mathrm{EP}\text{-}\mathrm{EC}\!\left(\rho_{2}^{\star}\right)\!-\!\vartheta_{1} P_{tot}\!\left(\rho_{2}^{\star}\right)\\
     & \leq \max\limits_{0< \rho < P_{max}} \left\{\mathrm{EP}\text{-}\mathrm{EC}\left(\rho\right)-\vartheta_{1} P_{tot}\left(\rho\right)\right\} = F(\vartheta_{1}).
    \end{aligned}
  \end{equation*}
  \vspace{-0.2cm}
    \par Thus, for any $0 \leq \vartheta_{1}<\vartheta_{2}$, we have $F(\vartheta_{2})< F(\vartheta_{1})$.
    \vspace{-0.1cm}
    \par Assume that $\rho_{t}^{\star}$ is the optimal transmit power that maximizes the function $F\big(q\vartheta_{1}+(1-q)\vartheta_{2}\big)$, where $\forall \ \vartheta_{1} \neq \vartheta_{2}$, and $0 \leq q \leq 1$. Then, based on the definition of convex function, we can obtain that
    \begin{equation*}
      \begin{aligned}
         & F\big(q\vartheta_{1}+(1-q)\vartheta_{2}\big) \\
         & = \max\limits_{0< \rho < P_{max}} \left\{\mathrm{EP}\text{-}\mathrm{EC}\left(\rho\right)-\big(q\vartheta_{1}+(1-q)\vartheta_{2}\big) P_{tot}\left(\rho\right)\right\}\\
         & = \quad \mathrm{EP}\text{-}\mathrm{EC}\left(\rho_{t}^{\star}\right)-\big(q\vartheta_{1}+(1-q)\vartheta_{2}\big) P_{tot}\left(\rho_{t}^{\star}\right)\\
         & = q\big(\mathrm{EP}\text{-}\mathrm{EC}\left(\rho_{t}^{\star}\right) - \vartheta_{1}P_{tot}\left(\rho_{t}^{\star} \right)\big)+ (1-q)\cdot \big(\mathrm{EP}\text{-}\mathrm{EC}\left(\rho_{t}^{\star}\right)\\
         &\quad \quad - \vartheta_{2}P_{tot}\left(\rho_{t}^{\star} \right)\big)\\
         & \leq q \cdot \max_{0\leq \rho \leq P_{max}}\big\{(\mathrm{EP}\text{-}\mathrm{EC}\left(\rho\right) - \vartheta_{1}P_{tot}\left(\rho \right)\big\}+ (1-q) \cdot \\
         & \quad \max_{0\leq \rho \leq P_{max}}\big\{\mathrm{EP}\text{-}\mathrm{EC}\left(\rho\right)- \vartheta_{2}P_{tot}\left(\rho\right)\big\}\\
         & = q F(\vartheta_{1}) + (1-q)F(\vartheta_{2}).
      \end{aligned}
    \end{equation*}
    \par Thus, for any $\vartheta_{1} \neq \vartheta_{2}$ and $0 \leq q \leq 1$, we can prove that $F\big(q\vartheta_{1}+(1-q)\vartheta_{2}\big) \leq q F(\vartheta_{1}) + (1-q)F(\vartheta_{2})$.
    \par So the proof of \textbf{\emph{Lemma 1}} is concluded.
\end{proof}

\begin{lemma}
  There exists one and only one optimal value of EP-EE $\vartheta^{\star}$ that satisfies the equation $F(\vartheta) = 0$.
\end{lemma}
\begin{proof}
   This lemma results from lemma 1 and the following facts. According to $\mathcal{P}7$, we have $F(0) \geq 0$ since $\mathrm{EP}\text{-}\mathrm{EC}\left(\rho\right) \geq 0$, and $P_{tot}\left(\rho_{t}^{\star} \right) > 0$ for any $0 \leq \rho \leq P_{max}$. In addition, $\lim\limits_{\vartheta \rightarrow +\infty}F(\vartheta) = -\infty$ and $\lim\limits_{\vartheta \rightarrow -\infty}F(\vartheta) = +\infty$.  Thus, there exists one and only one optimal value of EP-EE $\vartheta^{\star}$ that satisfies the equation $F(\vartheta) = 0$.
  \par So the proof of Lemma 2 is concluded.
\end{proof}

\begin{theorem}
  Assume that $\rho^{\star}$ is the optimal transmit power of problem $\mathcal{P}9$. Then, $\vartheta^{\star} = \frac{\mathrm{EP}\text{-}\mathrm{EC}\left(\rho^{\star}\right) }{P_{tot}\left(\rho^{\star}\right)} = \max\limits_{0< \rho < P_{max}} \left\{  \mathrm{EP}\text{-}\mathrm{EC}\left(\rho\right)-\vartheta P_{tot}\left(\rho\right) \right\}$ holds, if and only if $F(\vartheta^{\star}) = \max\limits_{0< \rho < P_{max}} \big\{\mathrm{EP}\text{-}\mathrm{EC}\left(\rho\right)$ $-\vartheta^{\star} P_{tot}\left(\rho\right)\big\} = 0$.
\end{theorem}

\begin{proof}
  \textbf{(i)}\ On the one hand, since $\rho^{\star}$ is the optimal transmit power of problem $\mathcal{P}8$, we can obtain that
  \begin{equation*}
    \vartheta^{\star} = \frac{\mathrm{EP}\text{-}\mathrm{EC}\left(\rho^{\star}\right) }{P_{tot}\left(\rho^{\star}\right)} \geq \frac{\mathrm{EP}\text{-}\mathrm{EC}\left(\rho\right)}{P_{tot}\left(\rho\right)},\ \forall \  0 \leq \rho \leq 0.
  \end{equation*}
  \par Based on the above inequality, we can obtain that
  \vspace{-0.2cm}
  \begin{subequations}\label{e57}
    \begin{align}
       & \mathrm{EP}\text{-}\mathrm{EC}\left(\rho\right) - \vartheta^{\star}P_{tot}\left(\rho\right) \leq 0,\\
       & \mathrm{EP}\text{-}\mathrm{EC}\left(\rho^{\star}\right) - \vartheta^{\star}P_{tot}\left(\rho^{\star}\right) = 0.
    \end{align}
  \end{subequations}
  \vspace{-0.4cm}
  \par From (\ref{e57}), we can observe that $F(\vartheta^{\star}) = \max\limits_{0< \rho < P_{max}} \left\{\mathrm{EP}\text{-}\mathrm{EC}\left(\rho\right)-\vartheta^{\star}P_{tot}\left(\rho\right)\right\} = 0$, and the maximum value of $F(\vartheta^{\star})$ is taken on the point of optimal transmit power $\rho^{\star}$.
  \par \textbf{(ii)}\ On the other hand, since $\rho^{\star}$ is the optimal transmit power of problem $\mathcal{P}8$, we can obtain that $\mathrm{EP}\text{-}\mathrm{EC}\left(\rho^{\star}\right)-\vartheta^{\star}P_{tot}\left(\rho^{\star}\right)=0$, which implies
  \begin{equation*}
    \mathrm{EP}\text{-}\mathrm{EC}\left(\rho\right)-\vartheta^{\star}P_{tot}\left(\rho\right) \leq \mathrm{EP}\text{-}\mathrm{EC}\left(\rho^{\star}\right)-\vartheta^{\star}P_{tot}\left(\rho^{\star}\right) = 0.
  \end{equation*}
  \par Based on the above inequality, we can derive the same conclusion as stated in (\ref{e57}a) and (\ref{e57}b). Thus, $F(\vartheta^{\star}) = \max\limits_{0< \rho < P_{max}} \left\{\mathrm{EP}\text{-}\mathrm{EC}\left(\rho\right)-\vartheta^{\star}P_{tot}\left(\rho\right)\right\} = 0$ and the maximum value of $F(\vartheta^{\star})$ is also taken on the point of optimal transmit power $\rho^{\star}$.
  \par So the proof of \textbf{\emph{Theorem 7}} is concluded.
\end{proof}
\vspace{-0.1cm}
\par According to Lemma 1, Lemma 2, and Theorem 7, solving the original problem $\mathcal{P}6$ is equivalent to solving the equation $F(\vartheta) = 0$, and we can obtain that
\vspace{-1em}
\begin{equation}\label{e58}
     \!\! \max_{0< \rho < P_{max}} \!\! \left\{  \mathrm{EP}\text{-}\mathrm{EC}\!\left(\rho\right)\!-\!\vartheta^{\star} P_{tot}\!\left(\rho\right) \right\} \!=\! 0 \Longleftrightarrow \vartheta = \vartheta^{\star}.
\end{equation}
\vspace{-2em}
\par Furthermore, according to Newton's iterative method\cite{dembo1982inexact}, the iterative function of the EP-EE can be derived as follows:
\vspace{-0.2cm}
\begin{equation}\label{e59}
   \begin{aligned}
     \vartheta(n+1) = \vartheta(n) - \frac{F(\vartheta(n))}{F^{'}(\vartheta(n))} = \frac{\mathrm{EP}\text{-}\mathrm{EC}(\rho^{\star}(n))}{P_{tot}(\rho^{\star}(n))},
   \end{aligned}
\end{equation}
where $\rho^{\star}(n)$ denotes the unique optimal transmit power of $\mathcal{P}8$ for a given value of EP-EE $\vartheta(n)$, and $(n)$ denotes the $n$-th iteration.
\par Therefore, we finally can obtain the following optimization problem as follows:
\vspace{-0.1cm}
\begin{subequations}\label{e60}
    \begin{align}
      \mathcal{P}10:\ \ &F(\rho^{\star}(n)|\vartheta(n))= \mathrm{EP}\text{-}\mathrm{EC}(\rho^{\star}(n)) - \vartheta(n) P_{tot}(\rho^{\star}(n)) = 0,\\
      & \vartheta(n+1) = \frac{\mathrm{EP}\text{-}\mathrm{EC}\left(\rho^{\star}(n)\right)}{P_{tot}\left(\rho^{\star}(n)\right)}.
    \end{align}
\end{subequations}

\par To solve $\mathcal{P}10$ effectively, we propose a novel and low-complexity algorithm named the outer-descent inner-search collaborative algorithm (ODISC). Algorithm 3 includes detailed descriptions of the ODISC. The core design idea of ODISC can be divided into two aspects for elaboration. On the one hand, in the inner-search part of ODISC, the Golden-Section search method (GSS) is used to search for the optimal transmit power $\rho^{\star}(n)$ for a given $\vartheta(n)$. Algorithm 2 provides detailed descriptions of the inner-search part of ODISC. On the other hand, to avoid tedious derivative operations, a concise closed-form iterative expression is also derived in the outer-descent part of ODISC by exploiting the derived iterative function (\ref{e59}), which speeds up the convergence of the proposed ODISC to the globally optimal solution.

\begin{algorithm}[h]
\caption{GSS Algorithm for Solving Equation (\ref{e60}a)}
\begin{algorithmic}[1]
\REQUIRE
Lower search bound $loB(n)$; Upper search bound $upB(n)$; Upper bound EP-EE $F_{1}(n)$; Lower bound EP-EE $F_{2}(n)$; Inner-search Convergence criterion $\epsilon_{1}$; Golden-section ratio $\tau = 0.618$;\\
\ENSURE
Transmit power $\rho(n)$;\\
  \WHILE{$|upB(n)-loB(n)| \geq \epsilon_{1}$}
   \IF{$F_{1}(n) > F_{2}(n)$}
     \STATE Update $upB(n) \leftarrow \rho_{2}^{\star}(n)$; $\rho_{2}^{\star}(n) \leftarrow \rho_{1}^{\star}(n)$; $F_{2}(n) \leftarrow F_{1}(n)$;\\
     \STATE Update $\rho_{1}^{\star}(n) \leftarrow upB(n)-\tau \left(upB(n)-loB(n)\right)$ and $F_{1}(n) \leftarrow F\left(\rho_{1}^{\star}(n)|\vartheta^{\star}(n)\right)$;\\
   \ELSE
     \STATE Update $loB(n) \leftarrow \rho_{1}^{\star}(n)$; $\rho_{1}^{\star}(n) \leftarrow \rho_{2}^{\star}(n)$; $F_{1}(n) \leftarrow F_{2}(n)$;\\
     \STATE Update $\rho_{2}^{\star}(n) \leftarrow upB(n)+\tau * \left(upB(n)-loB(n)\right)$ and $F_{2}(n) \leftarrow F\left(\rho_{2}^{\star}(n)|\vartheta^{\star}(n)\right)$;\\
   \ENDIF
  \ENDWHILE
  \STATE \textbf{return} $\rho^{\star}(n) \leftarrow \frac{1}{2}\left(\rho_{1}(n)+\rho_{2}(n)\right)$;\\
\label{code:recentEnd}
\end{algorithmic}
\end{algorithm}

\vspace{-0.1cm}
\begin{algorithm}[h]
\setstretch{0.90}
\caption{ODISC Optimization Algorithm}
\begin{algorithmic}[1]
\REQUIRE 
Lower search bound $LowerB$; Upper search bound $UpperB$; Inner-search convergence criterion $\epsilon_{1}$; Outer-descent convergence criterion $\epsilon_{2}$; Golden-section ratio $\tau = 0.618$; max-iterations $\mathrm{Max}$-$\mathrm{iter}$;\\
\ENSURE 
Optimal E-EE $\vartheta^{\star}$; Optimal transmit power $\rho^{\star}$; The number of iterations $n$;\\
\STATE Initialize $n \leftarrow 1$; $\vartheta^{\star}(n) \leftarrow 0$; $loB(n) \leftarrow LowerB$; $upB(n) \leftarrow UpperB$;\\
\COMMENT{\texttt{\textbf{\underline{\underline{Loop 1}}}: Newton's Iterative Optimization}\ \ \ \ \ \ \ \ \ \ \ \ \ \ \ \ \ \ \ \ \ \ \ \ \ }
\WHILE{$n$ $<$ $\mathrm{Max}$-$\mathrm{iter}$}
\STATE Compute $\rho_{1}^{\star}(n) = upB(n) - \tau * \left(upB(n)-loB(n)\right)$;\\
\STATE Compute $\rho_{2}^{\star}(n) = loB(n) + \tau * \left(upB(n)-loB(n)\right)$;\\
\STATE Compute $F_{1}(n) = F\left(\rho_{1}^{\star}(n)|\vartheta^{\star}(n)\right)$ and $F_{2}(n) = F\left(\rho_{2}^{\star}(n)|\vartheta^{\star}(n)\right)$;\\
\COMMENT{\texttt{\textbf{\underline{\underline{Loop 2}}}: Golden-Section search method}\ \ \ \ \ \ \ \ \ \ \ \ \ \ \ \ \ \ \ \ \ \ \ \ \ }
  \STATE Execute Algorithm 2 to solve equation (60a) and obtain its unique solution $\rho(n)$;\\ 
  \IF{$F\left(\rho(n)|\vartheta^{\star}(n)\right)=0$ or $F\left(\rho(n)|\vartheta^{\star}(n)\right)<\epsilon_{2}$}
    \STATE Update $\rho^{\star} \leftarrow \rho(n)$;\\
    \STATE \textbf{return} $\vartheta^{\star}$; $\rho^{\star}$; $n$;\\
  \ELSE
    \STATE Update EP-EE $\vartheta^{\star}(n+1) \leftarrow \frac{EP-EC(\rho^{\star}(n))}{P_{tot}(\rho^{\star}(n))}$;\\
    \STATE Update $n \leftarrow n+1$;\\
  \ENDIF
\ENDWHILE
\label{code:recentEnd}
\end{algorithmic}
\end{algorithm}

\vspace{-0.5cm}
\subsection{Computational Complexity Analysis}
\vspace{-0.1cm}
\par In each iteration, ODISC executes Algorithm 2 with the EP-EE $\vartheta^{\star}(n-1)$ as the input to obtain $\rho^{\star}(n)$, and its computational complexity is $\mathcal{O}\big(\log(\frac{1}{\epsilon_{2}})\big)$ \cite{tsai2010golden}. The value of the EP-EE $\vartheta^{\star}(n)$ is then updated based on the iterative function (\ref{e59}) before being used as the input of Algorithm 2 again in the subsequent $(n+1)$-th iteration. The computational complexity of the outer-descent part of ODISC is $\mathcal{O}\big(\log\big(\frac{upB(n)-loB(n)}{\epsilon_{1}}\big)\big)$ for the $(n)$-th iteration\cite{d2018learning}. As a result, the overall computational complexity of ODISC is
$\mathcal{O}\big(\sum\limits_{n=1}^{N_{max}}(\log(\frac{1}{\epsilon_{2}})) \cdot \log\big(\frac{upB(n)-loB(n)}{\epsilon_{1}}\big) \big)$, where $\epsilon_{1}$ and $\epsilon_{2}$ are the respective outer-descent and inner-search convergence criterion, $N_{max}$ is the termination number of iterations, and $upB(n)$ and $loB(n)$ are the upper-search and lower-search bounds of Algorithm 2 for the $n$-$\mathrm{th}$ iteration, respectively.

\vspace{-0.1cm}
\section{Performance Evaluation}
\par In this section, we conduct intensive numerical simulations to validate and demonstrate the proposed statistical QoS provisioning analysis and performance optimization schemes in the xURLLC-enabled MU-MIMO wireless networks. Unless otherwise specified, the default simulation parameters are listed in Table I.
\vspace{-0.2cm}
\begin{table}[h]
 \renewcommand{\arraystretch}{1.0}
 \caption{Simulation Parameter Settings} 
 \label{table_example}
 \centering
 \resizebox{0.70\columnwidth}{!}
 {
 \begin{tabular}{|c|l|c|p{0.30cm}}
  \hline
  \bfseries Parameter &\bfseries Physical meaning &\bfseries Value\\
  \hline
  $d_{min}$ & Minimum distance & $35$ m\\
  \hline
  $d_{max}$ & Maximum distance & $95$ m\\
  \hline
  $\mu_{cp}$ & Constant path loss & $-12$ dB \\
  \hline
  $\alpha_{0}$ & Path loss factor & $2.5$ \\
  \hline
  $M$ & Number of mobile UEs & $12$ \\
  \hline
  $N_{0}$ & Noise power spectral density & $-90$ dBm/Hz\\
  \hline
  $t_{DE}$ & Length of time slot & $0.5$ ms\\
  \hline
  $\epsilon_{1}/\epsilon_{2}$ & Convergence criterions& $10^{-5}$\\
  \hline
  $\mathrm{Max}$-$\mathrm{iter}$ & Maximum iterations & $10^{2}$\\ 
  \hline
  $P_{c}$ & Constant circuit power & $0.5$ watts\\
  \hline
  $P_{max}$ & Upper bound of transmit power & $2$ watts\\
  \hline
  $\varphi$  & Power amplifier coefficient & 0.5\\
  \hline
  $LowerB$ & Lower search bound & $10^{-6}$ \\
  \hline
  $UpperB$ & Upper search bound & $P_{max}$ \\
  \hline
 \end{tabular}
 }
\end{table}
\begin{figure}[h]
\centering
  \subfigure[]{
  \includegraphics[scale=0.225]{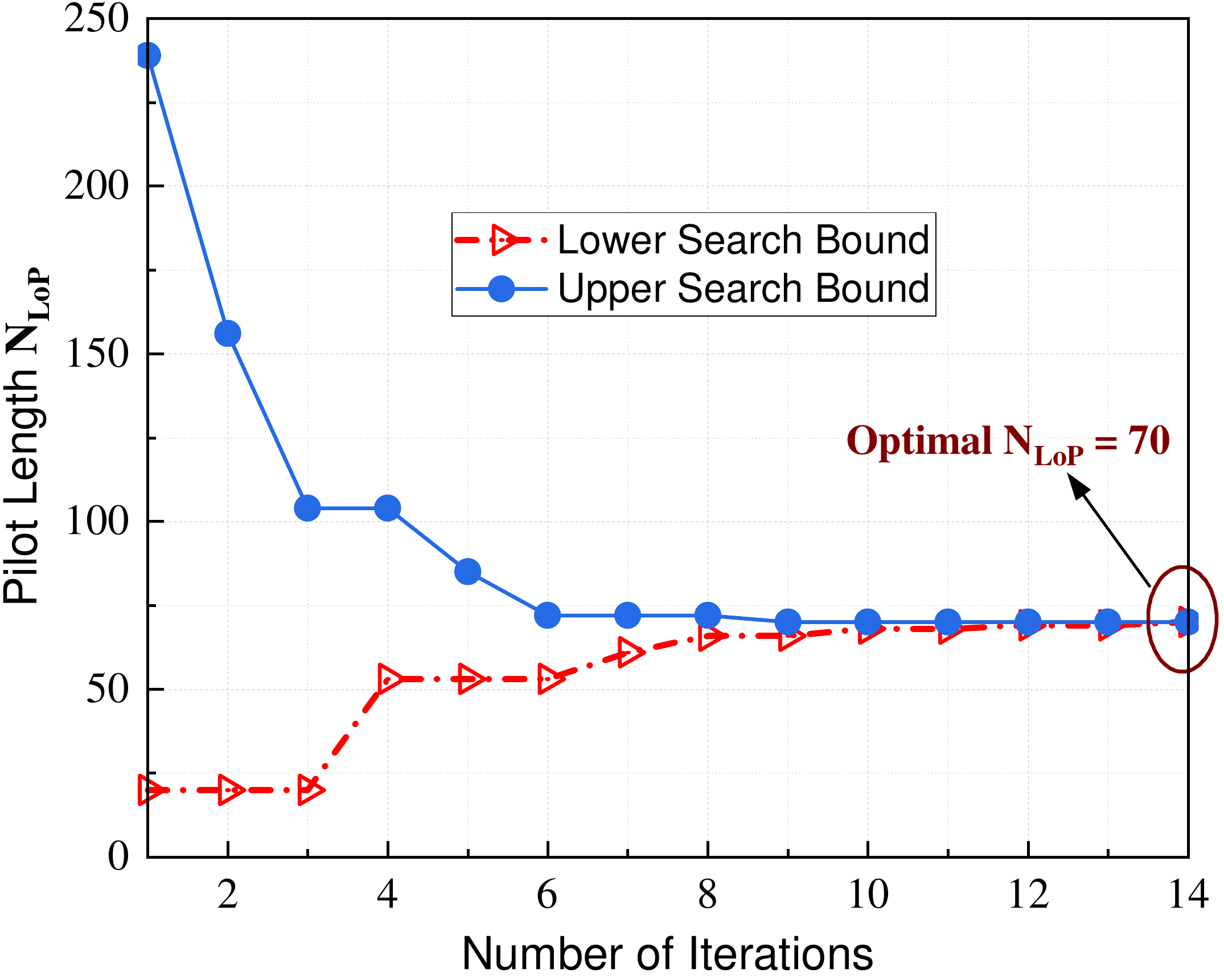}
  }
  \subfigure[]{
   \includegraphics[scale=0.225]{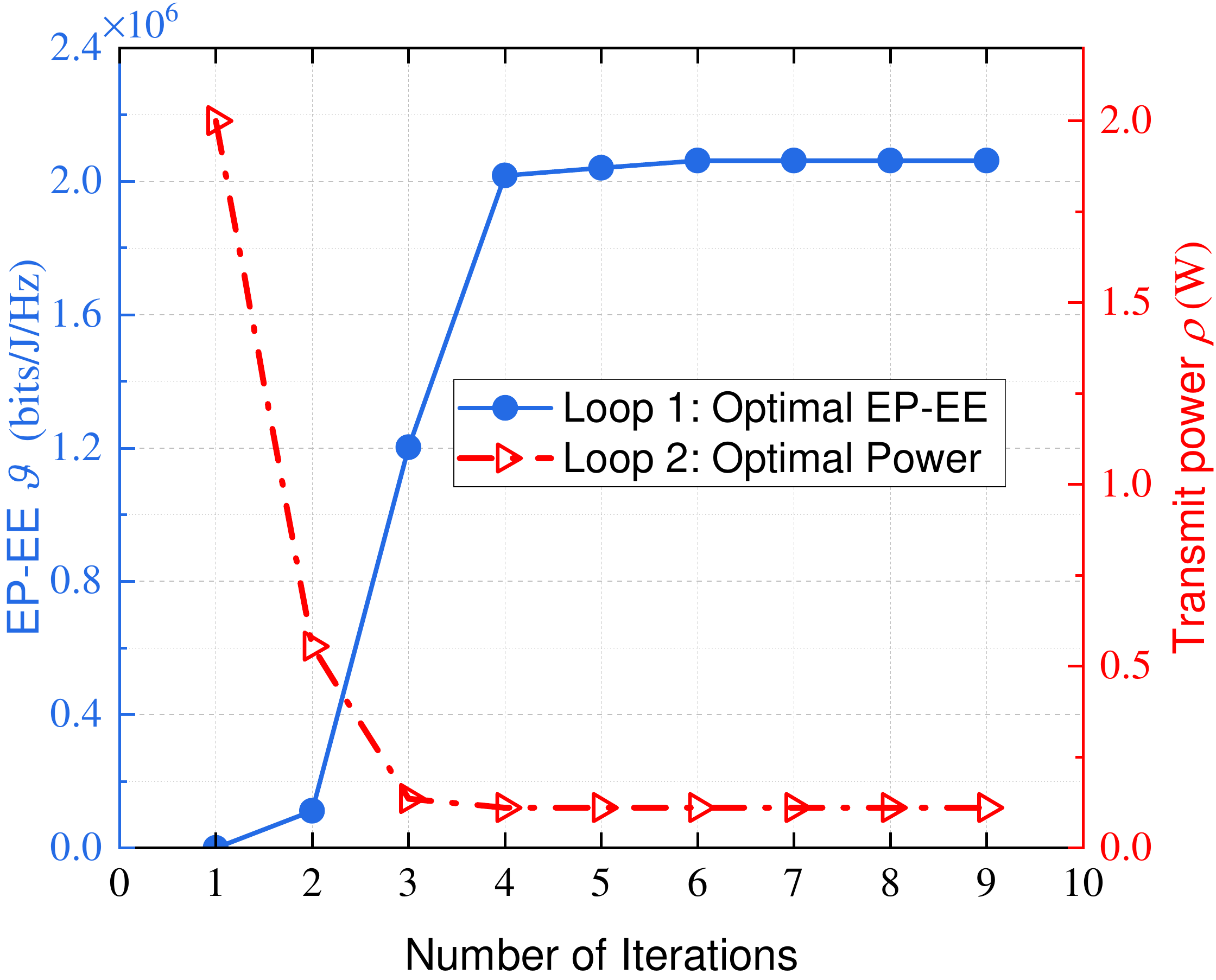}
  }
  \subfigure[]{
    \includegraphics[scale=0.225]{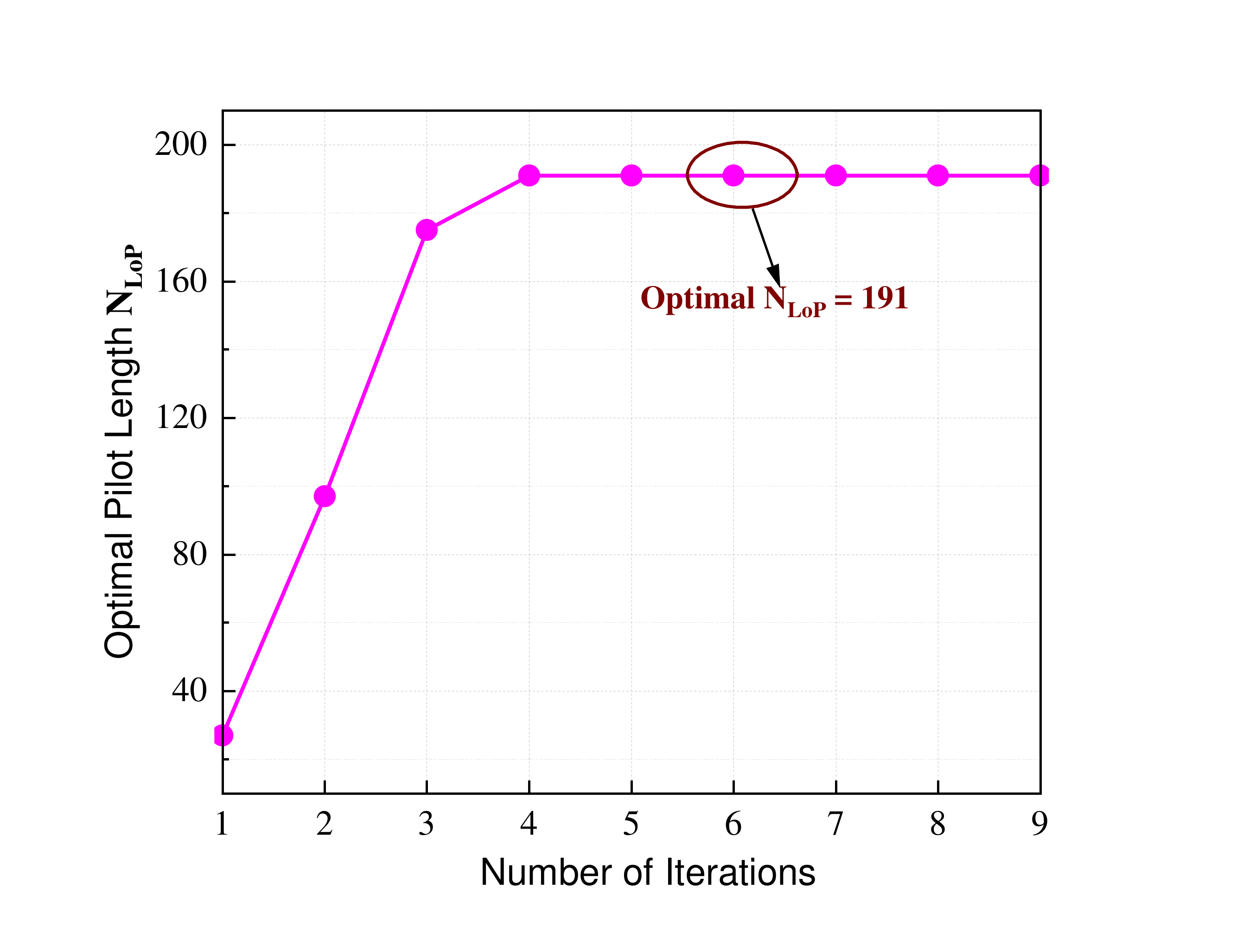}
  }
  \caption{(a) Convergence analysis of IFGSS Algorithm, where $\rho = 0.5$ W, $r_{m} = 0.2$ bpcu, $N_{T} = 50$, $B = 480$ kHz (i.e., $N_{CU} = 240$), $\theta_{m} = 0.2$; (b) convergence analysis of GSS Algorithm and ODISC Algorithm, and (c) the optimal pilot length $N_{LoP}$ in each iteration, where $N_{T} = 50$, $B = 10^{3}$ kHz (i.e., $N_{CU} = 500$), $r_{m} = 0.2$ bpcu, $\theta_{m} = 0.2$.}
\end{figure}
\vspace{-1em}
\subsection{Convergence Analysis of the Proposed Algorithms}
\vspace{-0.1cm}
\par Fig. 2 (a) depicts the convergence behavior of the upper-search and lower-search bounds of the proposed IFGSS algorithm. Numerical results reveal that the upper-search bound decreases drastically with the increasing number of iterations, while the lower-search bound increases. Moreover, both the upper-search and lower-search bounds simultaneously converge to the optimal pilot length in less than 14 iterations, and its value is 70. Consequently, we demonstrate effectiveness of Algorithm 1 with regards to its rapid convergence speed. Fig. 2 (b) and (c) depict the convergence behavior of the proposed GSS algorithm and the ODISC algorithm, respectively. From Fig. 2 (b), numerical results demonstrate that the value of EP-EE ameliorates considerably as the number of iterations increases; meanwhile, power consumption decreases dramatically. On the other hand, Fig. 2 (b) implies that the pilot overhead increases and eventually converges to 191 when optimizing EP-EE and power consumption, which is consistent with intuitive expectations since the reduction of transmit power entails an increase in pilot overhead. Additionally, numerical results in Fig. 2 (b) and (c) demonstrate that the proposed ODISC algorithm performs superiorly in terms of convergence speed since both EP-EE and transmit power converge simultaneously in less than 6 iterations. The principal cause of that is the collaborative optimization of EP-EE via outer-descent and inner-search design conventions, which enables the ODISC algorithm to rapidly attain optimal solutions. Specifically, the inner-search part degenerates the problem into a one-dimensional search, which is solved accurately by the proposed GSS algorithm. Additionally, the concise iterative function (\ref{e59}) provided in the outer-descent part eliminates the need for tedious derivative operations and accelerates the ODISC's convergence speed significantly.
\vspace{-1em}
\subsection{Tradeoff Between Decoding Error Probability and Pilot Length} 
\par In Fig. 3, the performance tradeoff between decoding error probability and pilot length is investigated. Numerical results indicate that both excessively short and excessively long pilot lengths result in an unsatisfactory decoding error probability. The objective analysis demonstrates that a too-short pilot length can fail to accurately reflect the CSI, while an overly long pilot length will inefficiently consume resources and reduce the effectiveness of data transmission accuracy. To corroborate the optimality of the proposed IFGSS algorithm, a comparison was made with EM. The results obtained from IFGSS and EM are marked with $\boldsymbol{\square}$ and $\boldsymbol{\ast}$ in Fig. 3, respectively. This comparison revealed that the outcomes of IFGSS are perfectly coincident with those obtained from EM, thus demonstrating that the proposed IFGSS algorithm can not only arrive at the optimal pilot length but also show superior convergence performance compared to EM, which traverses the entire search interval $\left[M,N_{CU}-1\right]$. Furthermore, numerical results also demonstrate that higher transmit power, increased BS antennas, and reduced maximum achievable data rate all have substantial positive effects on improving the decoding error probability and decreasing the pilot overhead. For instance, increasing the BS antennas from 30 to 50 drastically reduces the decoding error probability and pilot overhead from $1.56 \times 10^{-2}$ to $3.92 \times 10^{-5}$ and from $139$ to $69$, respectively.

\vspace{-1.5em}

\begin{figure}[h]
\begin{minipage}[t]{0.5\linewidth}
\centering
\includegraphics[scale=0.27]{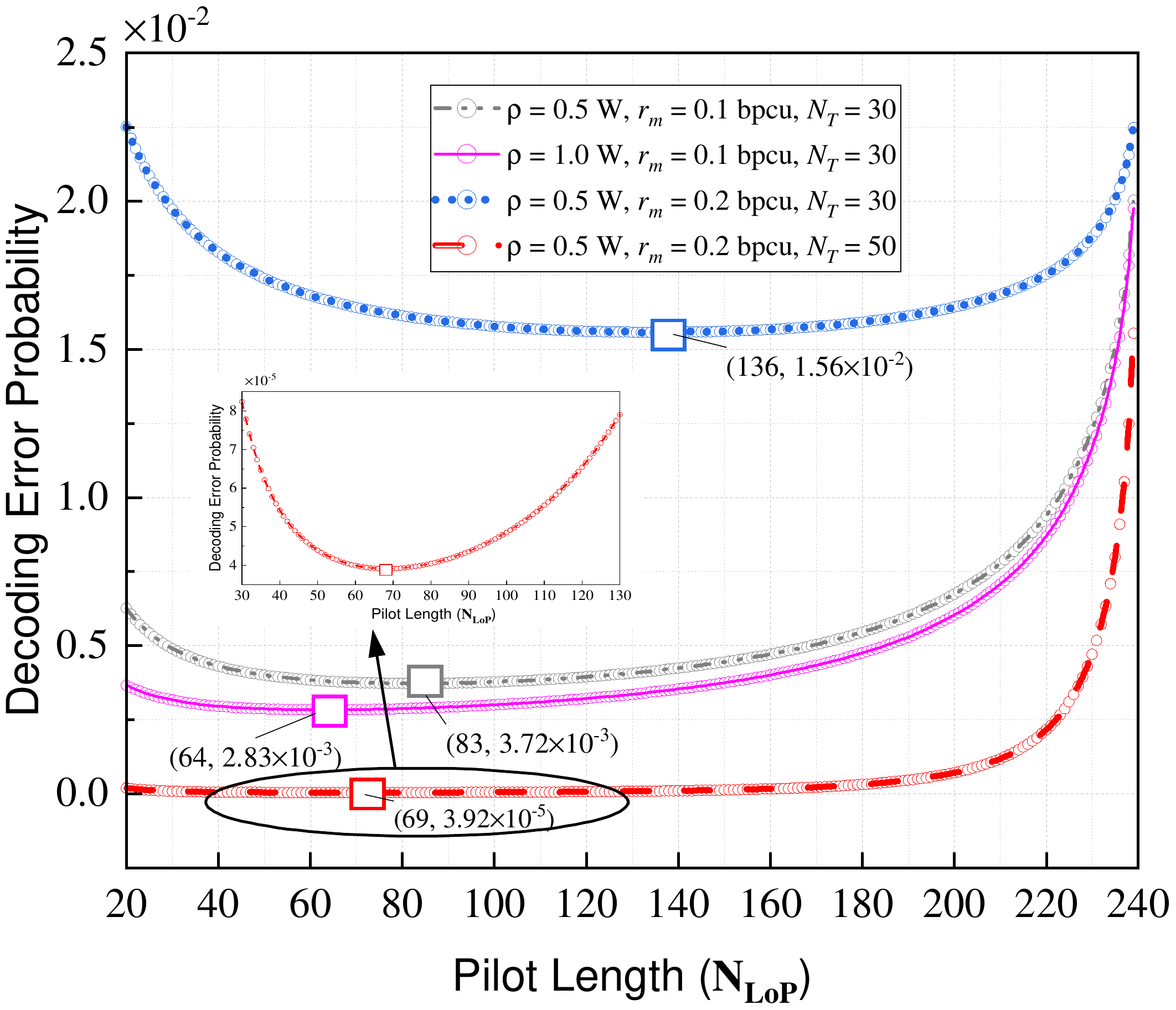}
\caption{The decoding EP of mobile UEs versus pilot length.$\ \ $ $B = 480$ kHz (i.e., $N_{CU} = 240$), $\theta_{m} = 0.2$.}
\label{frame}
\end{minipage}%
\begin{minipage}[t]{0.5\linewidth}
\centering
\includegraphics[scale=0.27]{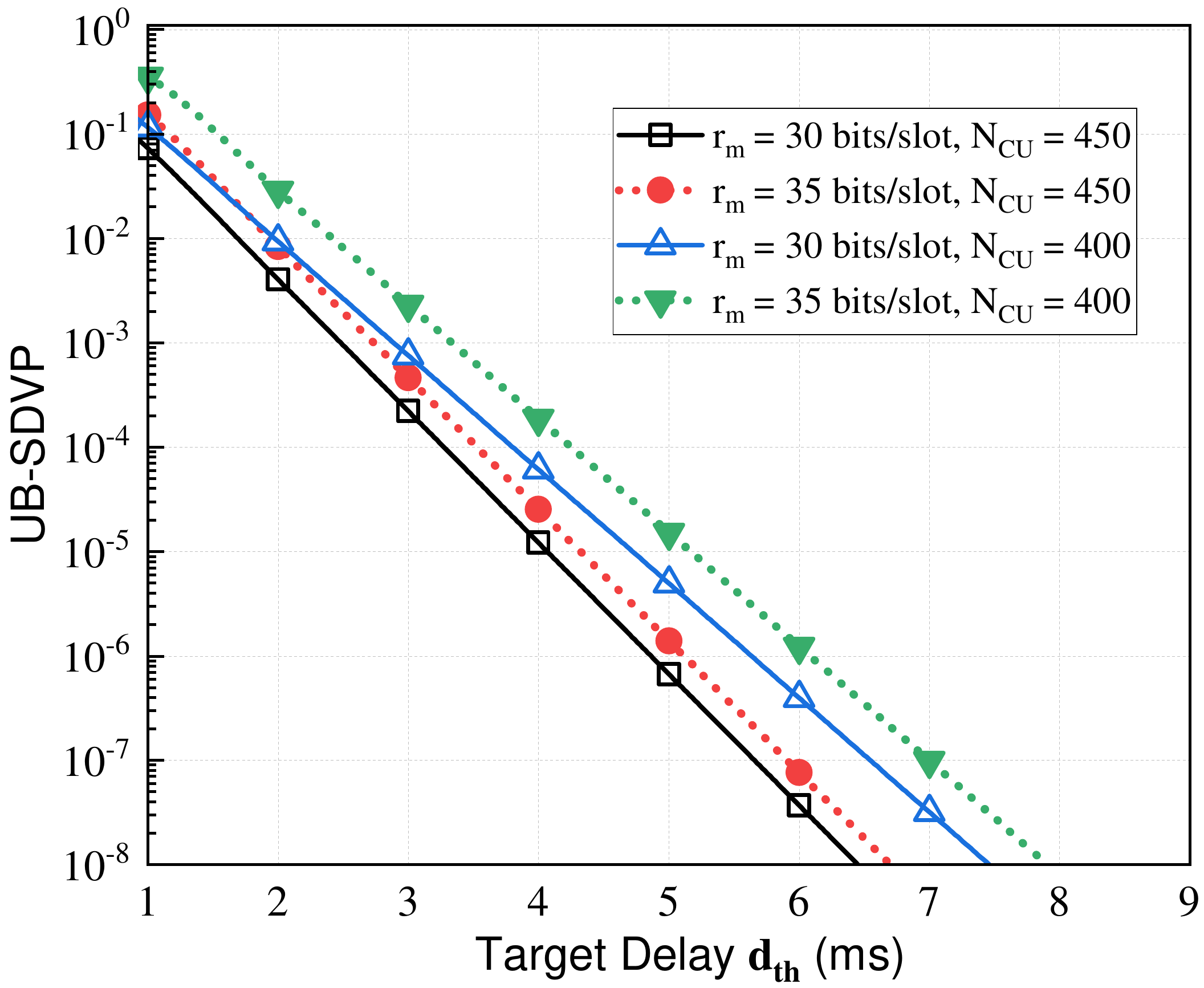}
\caption{The UB-SDVP versus target delay. $\rho = 0.5$ W, $N_{T} = 50$, $\epsilon_{m} = 10^{-6}$, $\theta_{m} = 0.2$.}
\label{label}
\end{minipage}
\end{figure}

\subsection{Tradeoff Between UB-SDVP and Target Delay}
\par Fig. 4 reveals the performance tradeoff between UB-SDVP and target delay. Numerical results show that the UB-SDVP decreases significantly as the target delay increases. If the target delay is stringent (for example, less than 1 ms), the UB-SDVP approaches an intolerable reliability of xURLLC, even exceeding $10^{-1}$. If the target delay is moderate (for example, $5\thicksim 7$ ms), the UB-SDVP can meet the extreme reliability of xURLLC overall, i.e., $1-10^{-7}$. If the target delay is relaxed (for example, greater than 7 ms), the UB-SDVP is less than $10^{-8}$. These experimental results suggest that loosening the target delay moderately can redeem remarkable reliability for the developed xURLLC-enabled massive MU-MIMO networks if conditions allow. Furthermore, the UB-SDVP can be improved by increasing the blocklength or decreasing the average arrival rate. This is primarily due to the fact that when the pilot length is determined, increasing the blocklength reserves additional resources for effective short-packet data communications, whereas decreasing the average arrival rate causes queues to build up more gradually, reducing the likelihood of short-packet data being stacked in the queues.

\vspace{-1.5em}
\subsection{Tradeoff Between UB-SDVP and Decoding Error Probability}
\vspace{-0.1cm}
\par According to FBC theory\cite{polyanskiy2010channel,polyanskiy2011feedback0,yang2014quasi}, short-packet data communications in the developed xURLLC-enabled massive MU-MIMO networks have a non-vanishing decoding error probability in the finite blocklength regime. Consequently, numerical simulations were conducted to evaluate the performance tradeoff between UB-SDVP and decoding error probability, as shown in Fig. 5. Numerical results suggest that the decoding error probability that minimizes the UB-SDVP lies between 0.01 and 0.02, and both excessively small and excessively large decoding error probabilities lead to an unsatisfactory UB-SDVP. Based on the objective analysis, it can be found that an excessively large decoding error probability implies that numerous short-packet data will be lost, thereby leading to a higher UB-SDVP, while an excessively small decoding error probability means that the system is forced to choose a relatively low service rate, which triggers that the short-packet data cached in the queues cannot be served in time. In addition, it can also be seen that increasing the transmit power and the number of BS antennas considerably enhances the reliability and mitigates the UB-SDVP caused by decoding error probability.
\vspace{-1em}
\subsection{Relation Between UB-SDVP and BS Antenna Numbers}
\vspace{-0.5em}
\par As illustrated in Fig. 6, we examine the relationship between UB-SDVP and BS antenna numbers for the developed xURLLC-enabled massive MU-MIMO networks in the finite blocklength regime. Numerical results show that increasing the number of BS antennas reduces the UB-SDVP significantly. This is primarily due to the fact that additional BS antennas furnish more spatial degrees of freedom for resource allocation as well as improved channel conditions (i.e., SINR), thus considerably lowering the UB-SDVP. Shadow fading is highly detrimental to wireless networks and can cause deep fading during xURLLC short-packet data communications. Nevertheless, our numerical results demonstrate that the developed xURLLC-enabled massive MU-MIMO networks can effectively compensate for the total throughput by equipping more BS antennas even when shadow fading is relatively large (e.g., 6 dB). Despite this, shadow fading persists as one of the primary factors contributing to the drastic deterioration in wireless network reliability. As a result, developing effective schemes to mitigate the adverse effects of shadow fading is crucial.

\vspace{-1em}

\begin{figure}[h]
\begin{minipage}[h]{0.5\linewidth}
\centering
\includegraphics[scale=0.27]{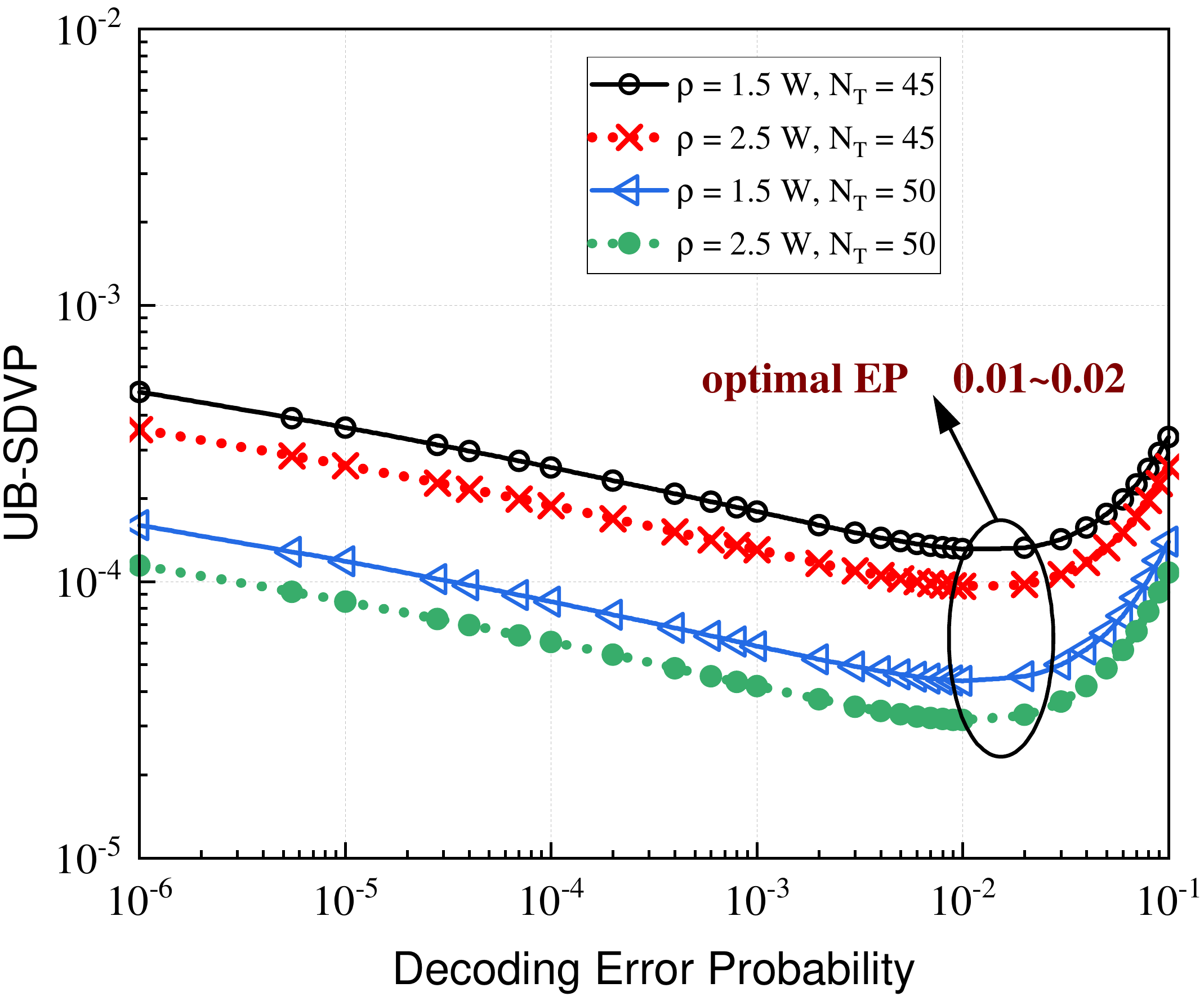}
\caption{The UB-SDVP versus decoding EP. $B = 1$MHz$\ \ $ (i.e.,$N_{CU} = 500$), $d_{th} = 5$ ms, $\theta_{m} = 0.2$, $\lambda_{m}^{\dag} = 40$ bits/slot.}
\label{frame}
\end{minipage}%
\begin{minipage}[h]{0.5\linewidth}
\centering
 \includegraphics[scale=0.27]{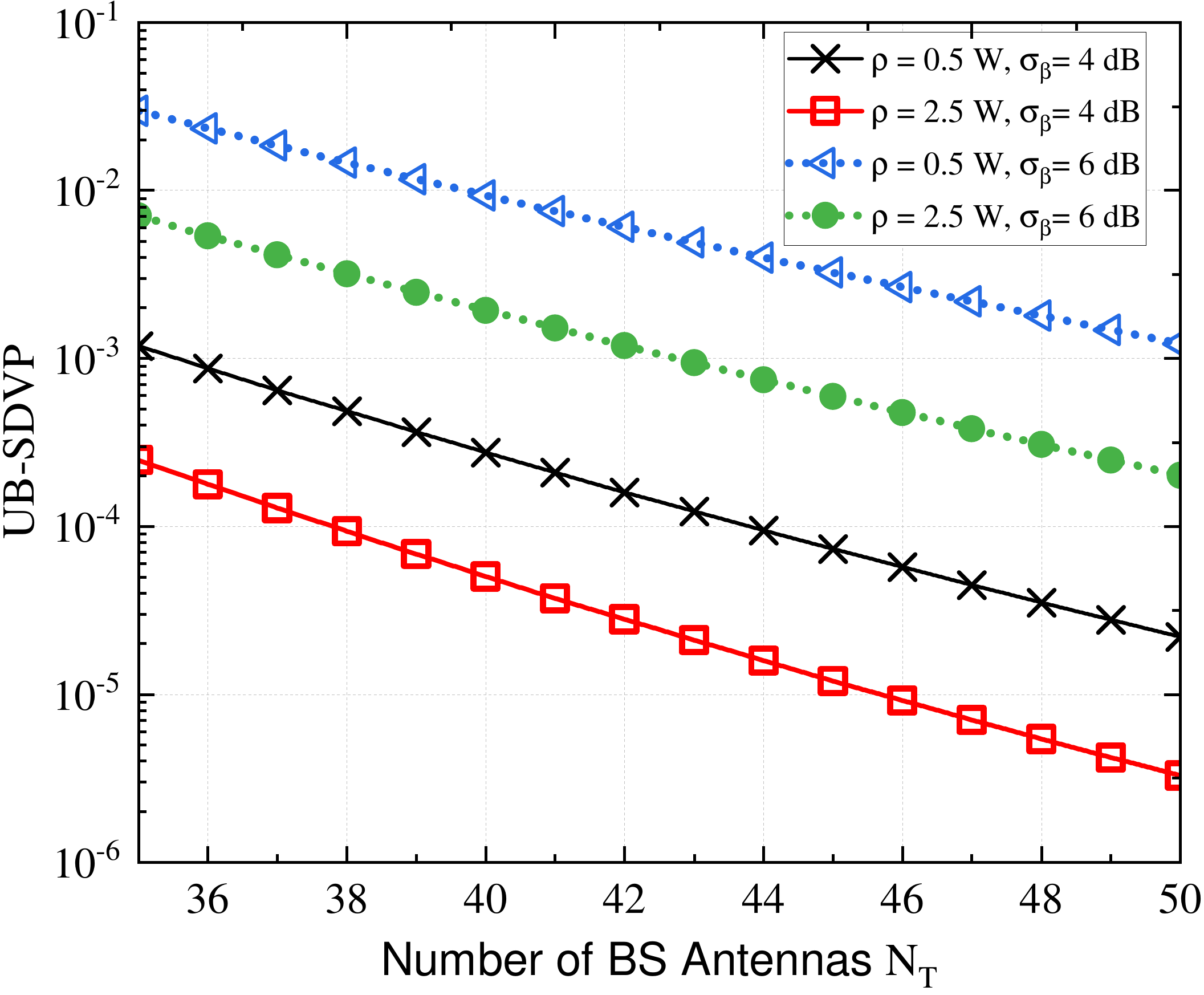}
 \caption{The UB-SDVP versus number of BS antennas. $B = 1$ MHz (i.e., $N_{CU} = 500$), $d_{th}=5$ ms, $\epsilon_{m} = 10^{-6}$, $\theta_{m} = 0.2$, $\lambda_{m}^{\dag} = 40$ bits/slot.}
\label{label}
\end{minipage}
\end{figure}

\vspace{-2.5em}

\subsection{Maximum EP-EC Versus Blocklength and QoS Requirements}
\vspace{-0.1cm}
\par In Fig. 7, we analyze the maximum EP-EC with various blocklengths and QoS exponents for proposed xURLLC-enabled massive MU-MIMO networks in the finite blocklength regime. As expected, the maximum EP-EC increases with the blocklength, whereas it decreases with more stringent QoS exponents, as shown in Fig. 7 (a). Additionally, Fig. 7 (a) implies that increasing the number of BS antennas is an effective means to serve xURLLC with provide more stringent QoS exponents. Besides, some intriguing observations can be discovered from Fig. 7 (b) and (c). On the one hand, when the QoS exponents are comparatively relaxed (i.e., $\theta_{m} \rightarrow 0$), the maximum EP-EC reaches its upper bound, as shown in Fig. 7 (b). For instance, when $\theta_{m} = 10^{-3}$, the maximum EP-EC demonstrates a tremendous improvement as the blocklength grows. This is because of the abundant resources can be reserved for the xURLLC short-packet data communications. In contrast, when the QoS exponents are stringent (i.e., $\theta_{m} \rightarrow \infty$), a lower bound of the maximum EP-EC is recorded, as demonstrated in Fig. 7 (c). For instance, when $\theta_ {m} = 0.8$, the maximum EP-EC is hardly improved even though the blocklength increases. Because of more stringent QoS exponents, despite adequate resources fail to meet the target. On the other hand, moderately slackening QoS exponents (e.g., $\theta_{m}:10\rightarrow 0.1$) significantly raises the maximum EP-EC, as shown in Fig. 7 (c). However, further slackening QoS exponents (e.g., $\theta_{m}:0.1 \rightarrow 0.001$) leads to an insignificant improvement in the maximum EP-EC when the QoS exponents are originally loose. Under such circumstances, increasing the blocklength can make the maximum EP-EC obtain a stepped promotion.

\begin{figure*}[h] 
\centering
  \subfigure[]{
    \includegraphics[scale=0.22]{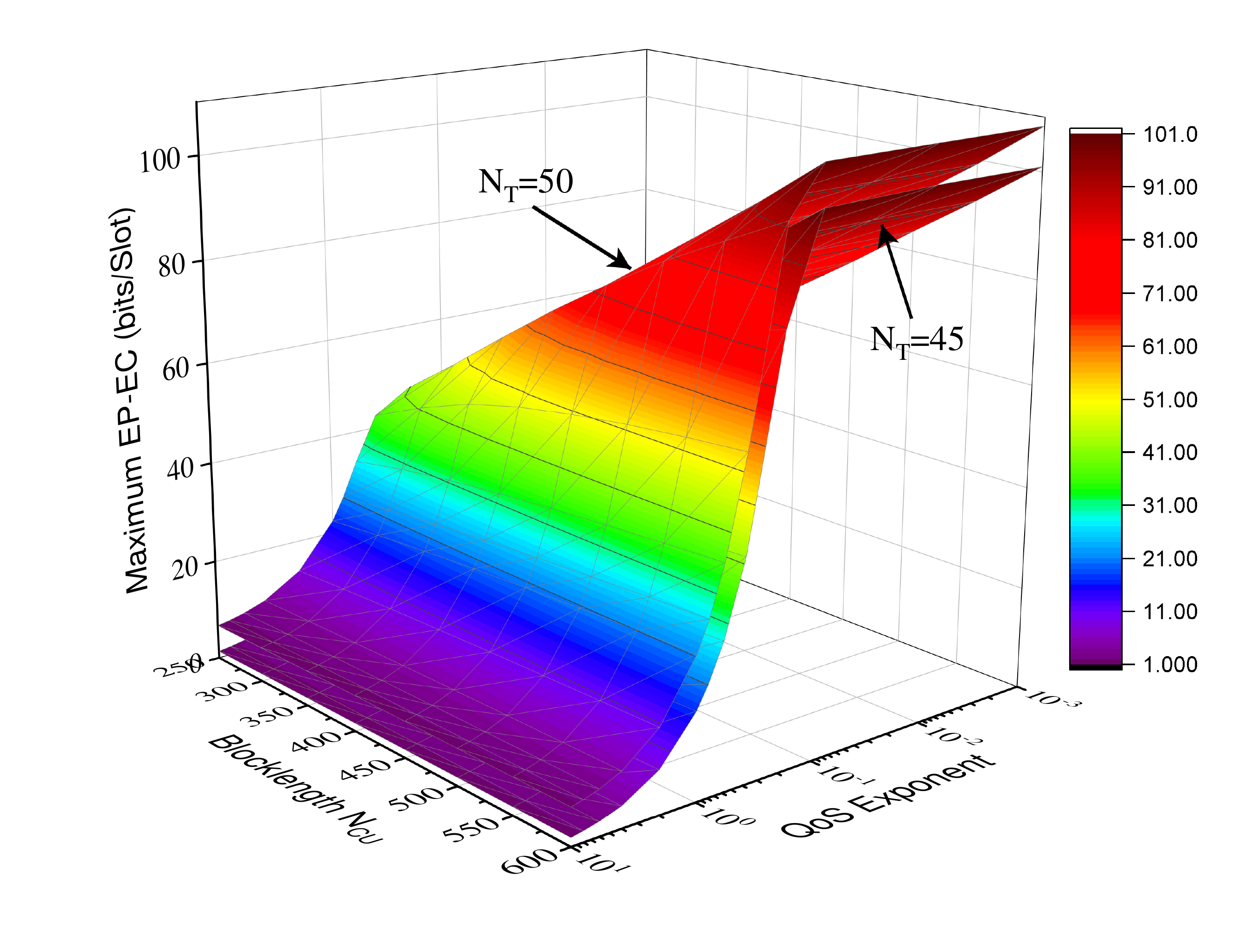}
  }
  \subfigure[]{
    \includegraphics[scale=0.23]{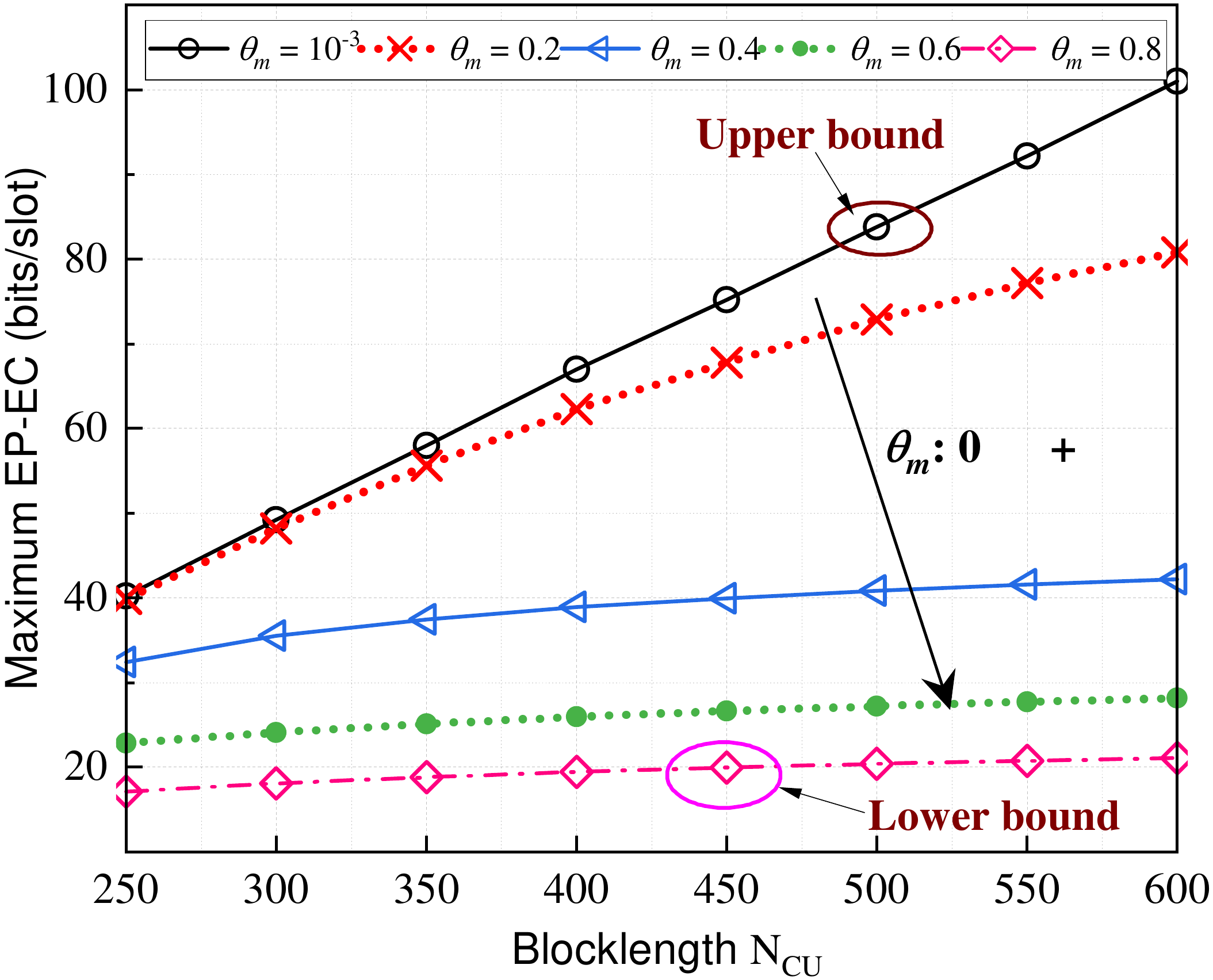}
  }
  \subfigure[]{
    \includegraphics[scale=0.23]{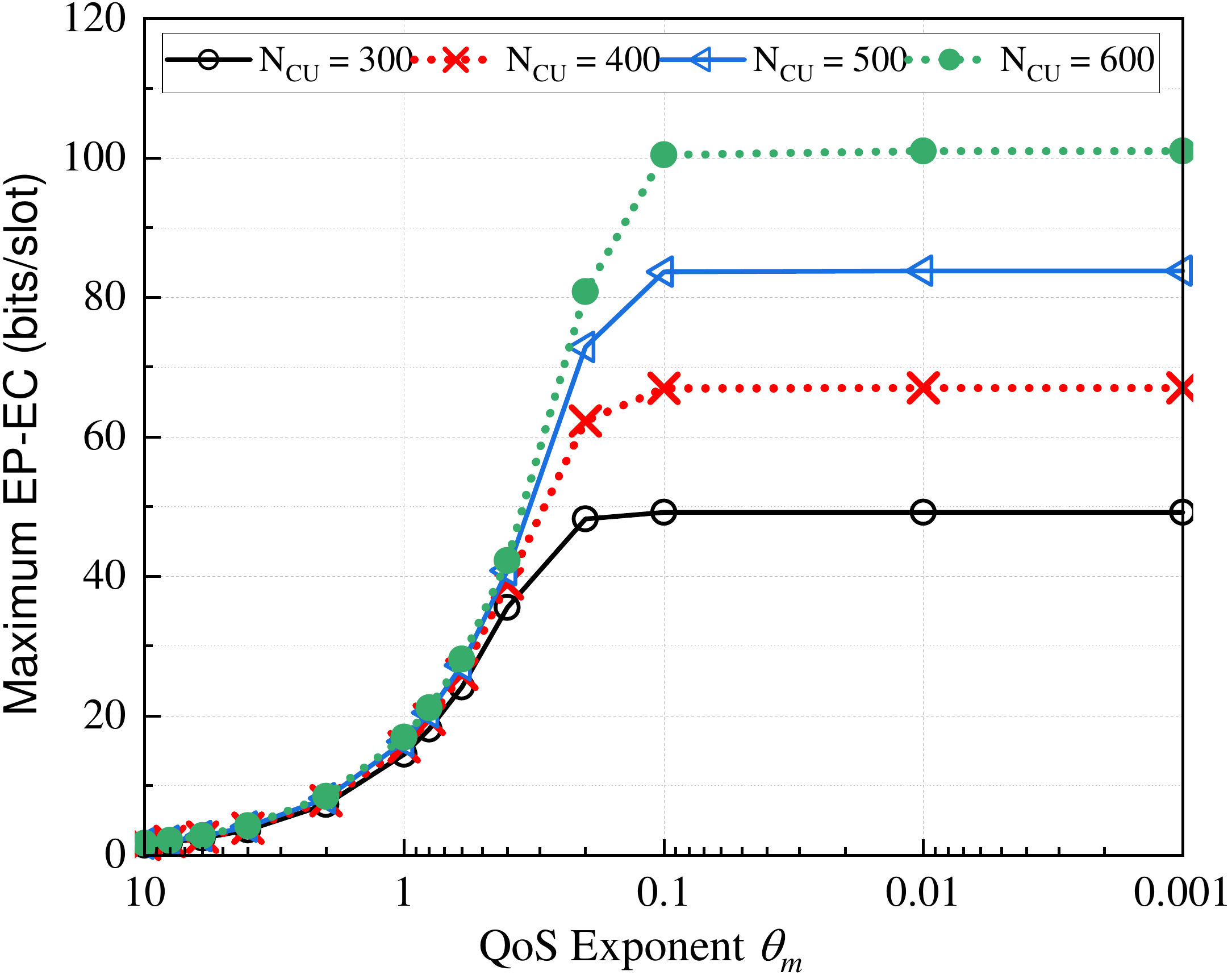}
  }
\caption{Maximum EP-EC versus Blocklength and QoS exponent. $\rho = 1.5$ W, $r_{m} = 0.2$ bpcu. In (a) and (b), $N_{T} = 45$, $\sigma_{B} = 6$ dB.}
\label{fig:label}
\end{figure*}
\vspace{-2em}

\subsection{Tradeoff Between Maximum EP-EE and Maximum Achievable Data Rate}
\par Fig. 8 reveals the performance tradeoff between maximum EP-EE and maximum achievable data rate for the developed xURLLC-enabled massive MU-MIMO networks. Some interesting observations can be deduced from Fig. 8 (a) and (b). To begin with, increasing the maximum achievable data rate necessitates investing more transmit power. In the second place, investing more transmit power can remarkably improve the maximum EP-EE performance before the inflection point. However, the performance improvement tends to flatten out as the inflection point approaches. After crossing the inflection point, increasing the transmit power continuously does not further improve the maximum EP-EE, and it changes from rising to declining. Last but not least, as the BS antennas increase and the shadow fading decreases, the inflection point of the maximum EP-EE will shift to the upper right, thus enlarging the lifting space of xURLLC-enabled massive MU-MIMO wireless networks. Fig. 8 (c) reveals the performance tradeoff between maximum EP-EE and optimal transmit power. There is a unique pair of maximum EP-EE and optimal transmit power, denoted as $\left(\rho^{\star},\vartheta^{\star}\right)$ for a given maximum achievable data rate, demonstrating the optimality of the proposed ODISC algorithm. It can also be noted that maximum EP-EE is a quasic-concave function concerning the optimal transmit power. As a result, by appropriately adjusting the maximum achievable data rate, the globally optimal EP-EE can be obtained, which is marked with red-solid symbols in Fig. 8 (c).
\vspace{-1.0em}
\subsection{Maximum EP-EE Versus Blocklengths and QoS requirements}
\vspace{-0.1cm}
\par As illustrated in Fig. 9 (a) and (b), we investigate the maximum EP-EE and the optimal transmit power with different blocklengths and QoS exponents for xURLLC-enabled massive MU-MIMO networks in the finite blocklength regime, respectively. Our primary focus is on examining the performance of maximum EP-EE, for which the QoS exponents vary within a moderate range. By combining with Fig. 9 (a) and (b) for comparative analysis, we can observe that increasing the blocklength noticeably ameliorates the maximum EP-EE, and at the same time, ensures low power consumption when the QoS exponents are loose. The intuition behind this is that increasing blocklength is an economical way to enhance the maximum EP-EE performance as minimal extra power is needed when the QoS exponents are loose. When the QoS exponents are relatively stringent, however, increasing blocklength is not as cost-effective in terms of improving the maximum EP-EE. Moreover, numerical results indicate that adding more BS antennas has a positive impact on overall maximum EP-EE while only slightly increasing the overall power consumption. The intuition behind this is that additional transmit power budget is needed to maintain the working quality of the newly added antennas when the number of BS antennas increases. Nevertheless, we are confident that the increased power budget is worthwhile, as it is used to ensure the stability of the antenna array, which is critical factor in the successful improvement of maximum EP-EE by increasing BS antennas.
\begin{figure*}[h]
\centering
  \subfigure[]{
   \includegraphics[scale=0.23]{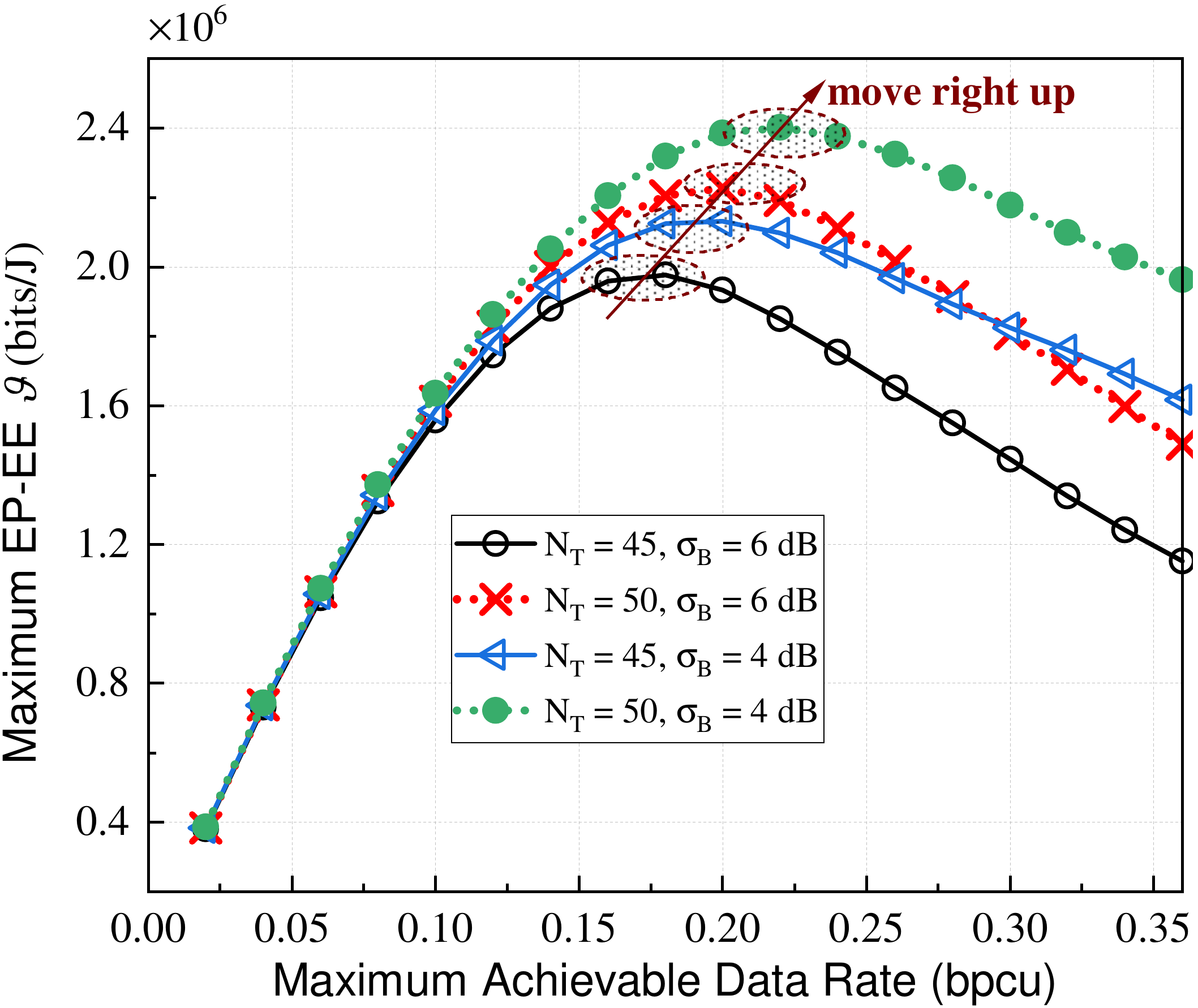}
  }
  \subfigure[]{
   \includegraphics[scale=0.23]{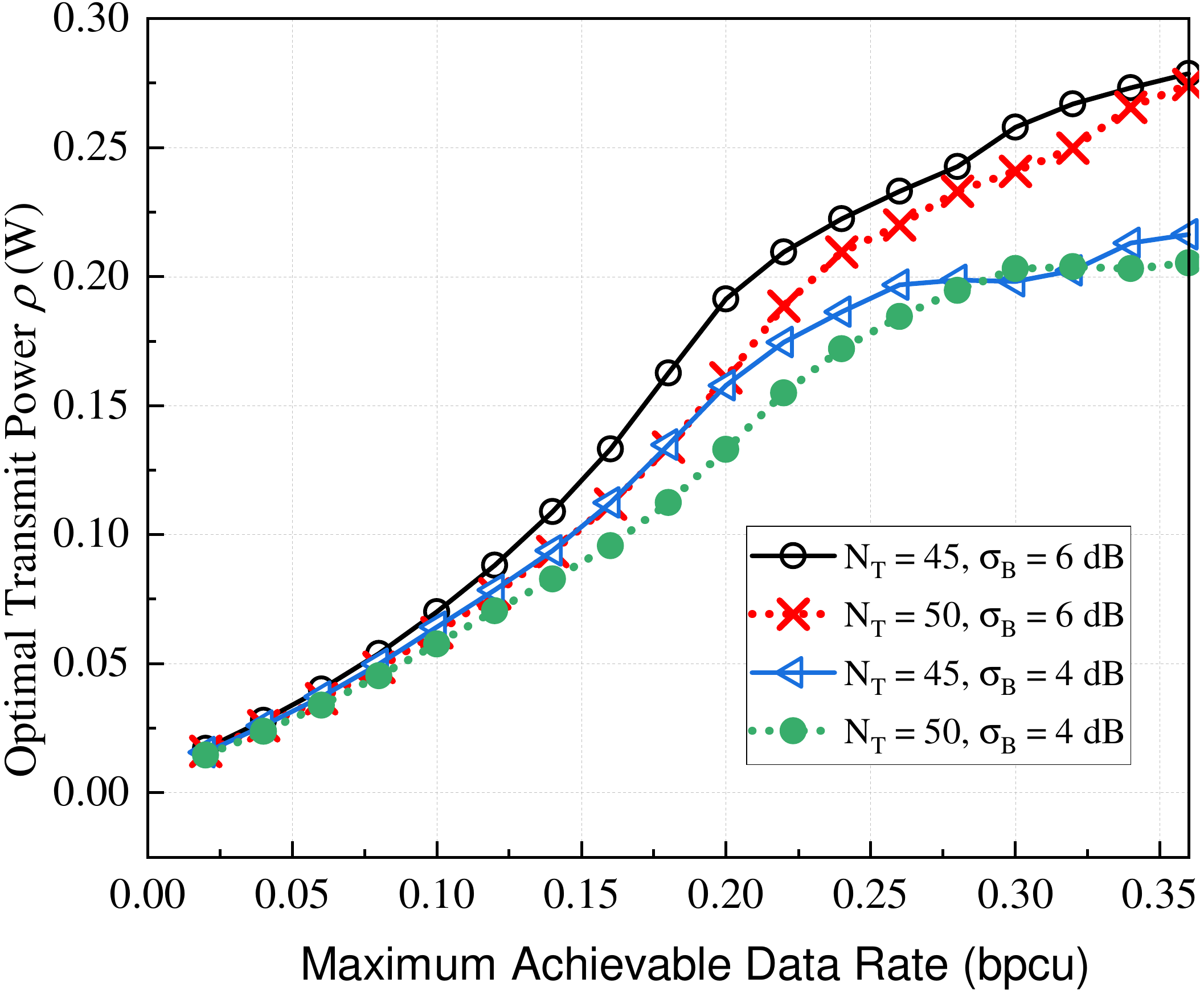}
  }
  \subfigure[]{
   \includegraphics[scale=0.23]{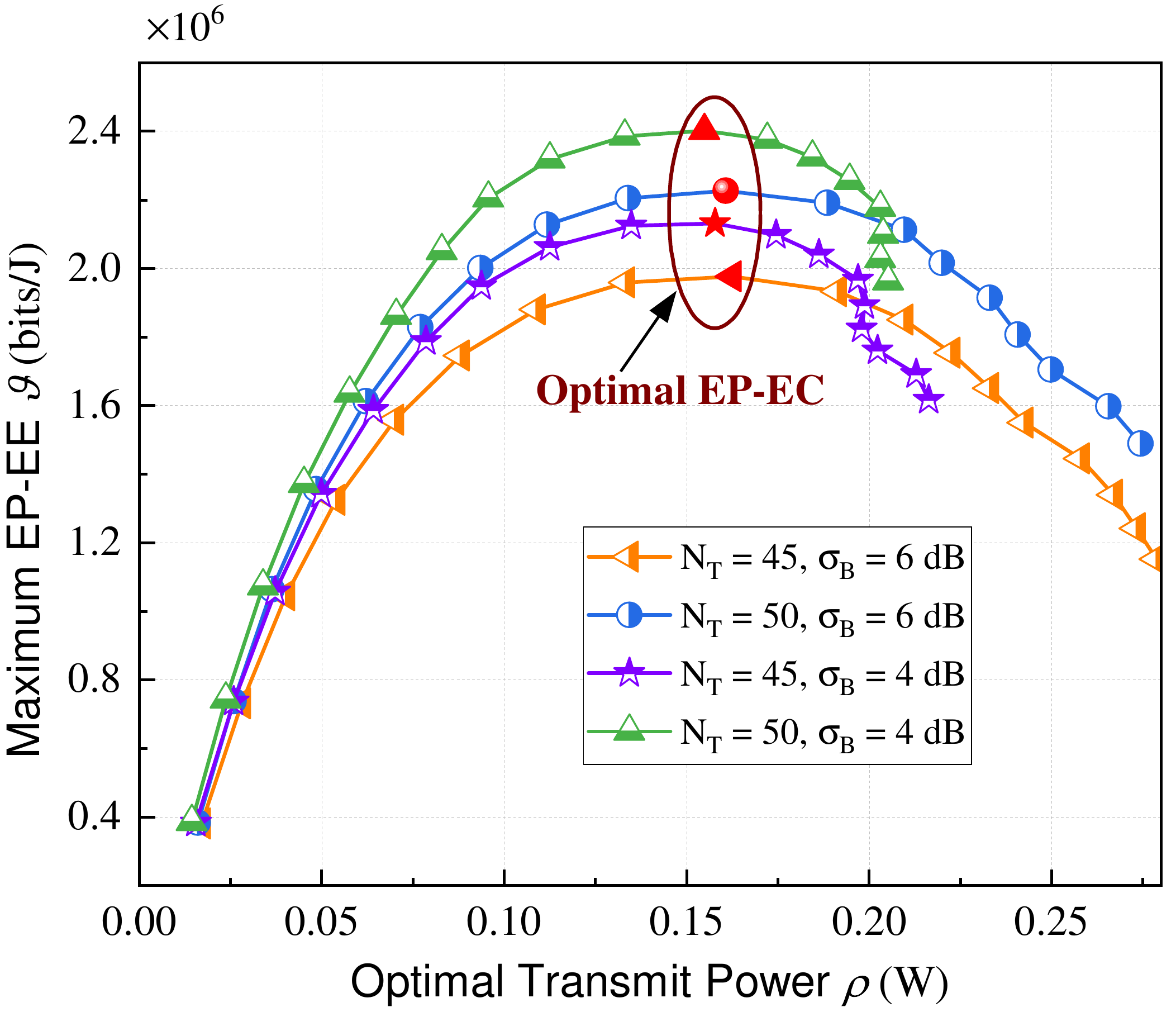}
  }
  \caption{(a) The maximum EP-EE with different maximum achievable data rate. (b) The optimal transmit power with different achievable data rate. (c) The maximum EP-EE versus the optimal transmit power. $B = 1$ MHz (i.e., $N_{CU} = 500$), $\theta_{m} = 0.2$.}
\end{figure*}

\begin{figure*}[htbp] 
\centering
  \subfigure[]{
    \includegraphics[scale=0.25]{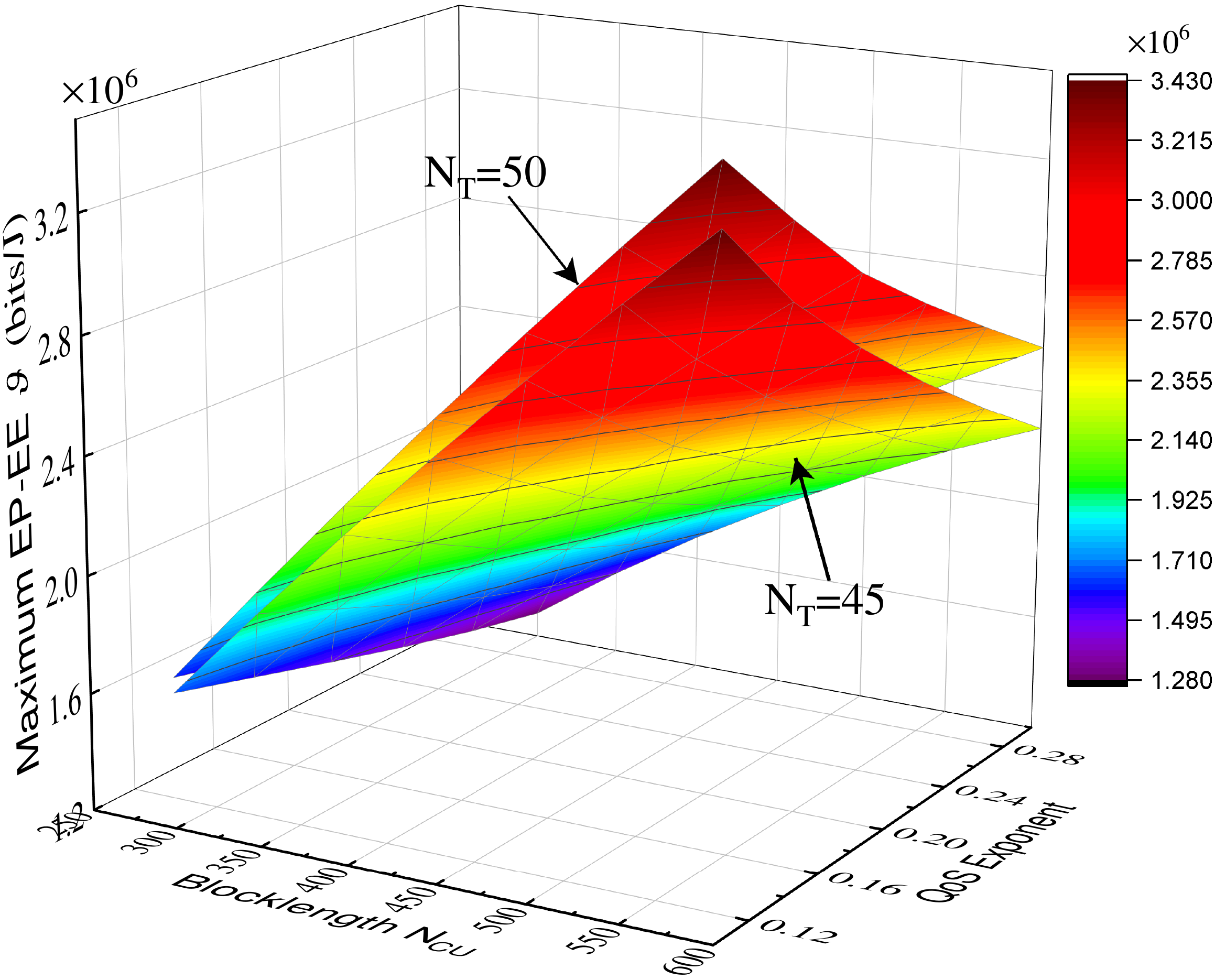}
  }
  \subfigure[]{
    \includegraphics[scale=0.25]{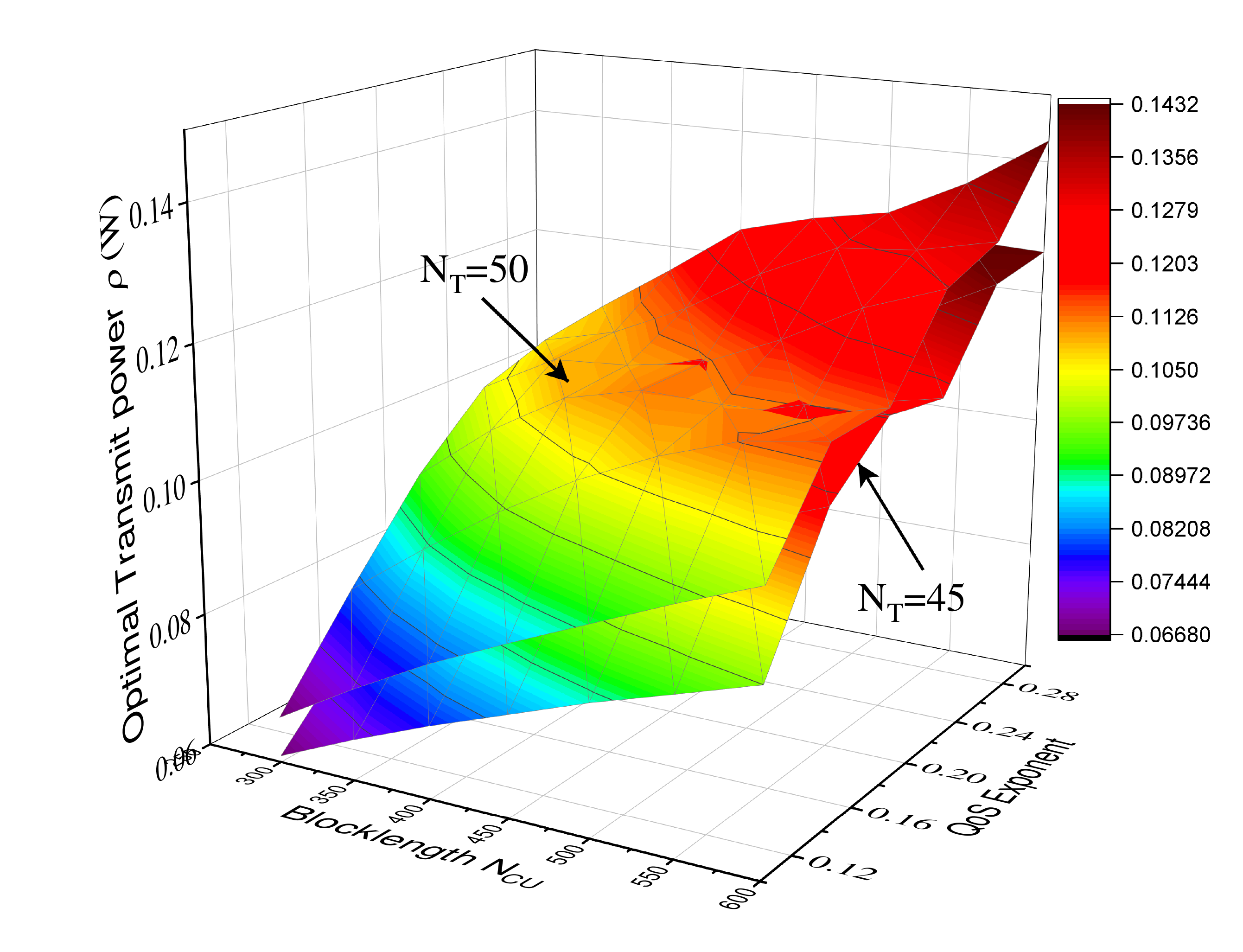}
  }
\caption{(a) The maximum EP-EE vs. blocklength and QoS exponents. (b) The optimal transmit power vs. blocklength and QoS exponents. $r_{m} = 0.2$ bpcu, $\sigma_{B} = 6$ dB.}
\label{fig:label}
\end{figure*}

\vspace{-0.4cm}
\section{Conclusion}
\par In this paper, we have revealed significant insights from the standpoint of SNC for statistical QoS provisioning analysis and performance optimization of xURLLC. In particular, we have developed xURLLC-enabled massive MU-MIMO networks, which can remarkably enhance reliability and comprehensively accommodate xURLLC features. We have then provided penetrative statistical QoS provisioning analysis for xURLLC by leveraging the promoted SNC theory. Additionally, two novel concepts, known as EP-EC and EP-EE, have been proposed to characterize the tail distribution and performance tradeoffs of xURLLC. Based on the proposed theoretical framework, UB-SDVP, EP-EC, and EP-EE optimization problems have been investigated and solved. Extensive simulations demonstrate that the proposed framework can considerably reduce computational complexity, while quantitatively providing various tradeoffs and optimization performance of xURLLC concerning UB-SDVP, EP, EP-EC, and EP-EE.

\par For future work, our proposed SNC-based theoretical framework will be extended to analyze the freshness (i.e., AoI) of xURLLC traffic, and the AoI-driven statistical QoS provisioning analysis and performance optimization for xURLLC will be considered. In addition, based on the lessons we learned in this research, an envisaged massive xURLLC deployment will be demonstrated with experiments by employing the high data rate millimeter wave technology and the latest cross-layer optimization strategies.
\vspace{-0.5em}
\footnotesize
\bibliographystyle{IEEEtran}
\bibliography{IEEEabrv,ref}

\end{document}